\documentclass[journal]{IEEEtran}
\usepackage{amsmath,nccmath}
\usepackage{amssymb}
\usepackage{graphicx}
\usepackage{esint}
\usepackage{amsmath,amssymb,amsthm}
\usepackage{thmtools,thm-restate}
\usepackage{hyperref}
\usepackage{cleveref}
\usepackage{ifpdf}
\usepackage{cite}
\usepackage{color}
\usepackage{placeins}
\usepackage{float}
\usepackage{tabularx}
\usepackage{colortbl}
\usepackage{pgfplots}
\usepackage{tikz}
\usepackage{tikzscale}
\pgfplotsset{compat=newest}
\usetikzlibrary{plotmarks}
\usetikzlibrary{arrows.meta}
\usepgfplotslibrary{patchplots}
\usepackage{grffile}
\pgfplotsset{plot coordinates/math parser=false}
\newlength\figureheight
\newlength\figurewidth
\usepackage{pgfgantt}
\usepackage{pdflscape}
\begin{document}

\title{Defense Against Smart Invaders with Swarms of Sweeping Agents}

\author{Roee M. Francos
        and Alfred M. Bruckstein
\thanks{Roee M. Francos and Alfred M. Bruckstein are with the Faculty
of Computer Science, Technion- Israel Institute of Technology, Haifa, Israel, 32000, Emails:(roee.francos@cs.technion.ac.il,
  alfred.bruckstein@cs.technion.ac.il).}
}

\maketitle
\begin{abstract}
The goal of this research is to devise guaranteed defense policies that allow to protect a given region from the entrance of smart mobile invaders by detecting them using a team of defending agents equipped with identical line sensors. By designing cooperative defense strategies that ensure all invaders are detected, conditions on the defenders' speed are derived. Successful accomplishment of the defense task implies invaders with a known limit on their speed cannot slip past the defenders and enter the guarded region undetected. The desired outcome of the defense protocols is to defend the area and additionally to expand it as much as possible. Expansion becomes possible if the defenders' speed exceeds a critical speed that is necessary to only defend the initial region. We present results on the total search time, critical speeds and maximal expansion possible for two types of novel pincer-movement defense processes, circular and spiral, for any even number of defenders. The proposed spiral process allows to detect invaders at nearly the lowest theoretically optimal speed, and if this speed is exceeded, it also allows to expand the protected region almost to the maximal area.
\end{abstract}


\IEEEpeerreviewmaketitle

\section{Introduction}
\label{sec:introduction}
\IEEEPARstart{T}{he} objective of this paper is to develop efficient guaranteed defense search strategies in which a swarm of $n$ defending agents must guarantee the detection of an unknown number of smart invaders from entering a region which
the defenders guard.  An initially given circular region of radius $R_0$ is assumed not to contain mobile invaders at the beginning of the sweep protocol, and is referred to as the initial protected region. The invaders may attempt to move into the  protected region from any point outside of the initial protected region and try to enter  the protected region (the region where invaders are not located) at a maximal speed of $V_T$, known to the
defenders. All defenders move at a speed $V_s > V_T$ and detect the invaders with linear sensors of length $2r$.  Once a defender's sensor touches a particular location, potential invaders that might have been present there are
detected and therefore are ``eliminated". Each guaranteed defense strategy requires a minimal speed that depends on the trajectory of the sweeping defenders and imposes a lower bound on the speed of the defenders. This critical speed is derived to ensure the satisfaction of the guarding task. Increasing the speed above the lower bound enables the defending agents to not only complete the guarding task but also to expand the guarded region as well.

Performing an efficient defense protocol requires that the footprint of the defenders' sensors minimally overlaps the protected region, thus allowing them to detect invaders further away from the protected region and stop their advance. By sweeping around the initial region they are required to guard, defenders with speeds higher than the critical speed, can increase the protected region, up to a circular region with a maximal radius that depends on their excess speed, their sensing capabilities and the sweep strategy they employ. This research paper develops two guaranteed defense search protocols for a swarm consisting of an even number of defenders that sweep the region. There are two goals for each developed defense strategy, defending the initial protected region in the defense task and performing the maximal expansion task in which the defenders execute their defense strategy until the protected region reaches its maximal defendable area, by employing novel pincer-movement search strategies. The proposed defense protocols are based on pairs of defenders that move toward each other thus entrapping invaders and halting their advance into the protected region.
\subsection{Overview of Related Research}
Multi-agent search problems have been an active area of research for almost a century, where early works by \cite{koopman1980search} focused on designing algorithms for detecting ships and submarines from surveillance aircraft in the English channel during the second world war.  Multi-agent search tasks involve searching for static or mobile targets and can take place in environments that range from being fully or partially known to being completely unknown \cite{stone2016optimal,rekleitis2004limited,alpern2006theory}. 

In case the targets being searched are static, searching the entire area in which the targets are located will surely result in their detection. Therefore, in such scenarios, the goal in designing an optimal searching algorithm is to find a traversal path for the searching agents that locates all targets in minimal time. In case the targets are mobile, detection is not always guaranteed since the targets' movements might prevent the searchers to detect them. This situation can occur even in closed and confined environments, in which the targets cannot exit the borders of the area being searched. In this paper we address the detection of a more challenging type of target, a smart mobile target that may perform evasive maneuvers by detecting and responding to the movements of the defending team in order to avoid being intercepted by the defenders that wish to prevent it from entering the protected region. Smart targets, which in the context of this paper are referred to as invaders, are assumed to have full knowledge of the defenders' strategy and to use that knowledge, to the best of their ability, in order to devise a counter strategy that allows them to enter the protected region without being caught.
In this work we are interested in developing guaranteed defense strategies against smart invaders, implying that regardless of the infiltration plan the invaders choose and their resulting trajectories, they will all be detected by the defending team. In \cite{vincent2004framework}, Vincent et. al investigate guaranteed detection of smart targets in a channel environment using a team of detecting sweeping agents and \cite{altshuler2008efficient} provides optimal strategies to the same problem.

In \cite{mcgee2006guaranteed}, McGee et al. study how to defend a given planar circular region against the entrance of smart intruders. The intruders do not have any maneuverability restrictions besides an upper limit on their speed. The defenders are equipped with sensors that detect intruders that are inside a disk shaped region around them. The considered search pattern is composed of spiral and linear sections. 
 
Somewhat related problems are pursuit-evasion games, in which the pursuers' goal is to detect and catch the evaders and the evaders goal is to avoid being detected and caught by the pursuing team. There are several variants of pursuit-evasion games which include different combinations of single and multiple evaders and pursuers settings. Pursuit-evasion games were also applied to address defending a region from the entrance of intruders. Such works are \cite{shishika2018local,shishika2019team,shishika2020cooperative} by Shishika et al. which investigate perimeter defense games and emphasize the cooperation between pursuers to improve the defense tactic. In \cite{shishika2018local}, members of the defending team of agents cooperate and form  defender pairs by moving in a "pincer movement"  to prevent intruders from entering a convex region in the plane. Cooperation between the defender sub-teams, allows to extend the winning regions of the defender team compared to performing uncooperative defense strategies. 

In \cite{makkapati2019optimal}, pursuit–evasion problems involving multiple pursuers and multiple evaders (MPME) are studied. Pursuers and evaders are all assumed to be identical, and pursuers follow either a constant bearing or a pure pursuit strategy. The problem is simplified by adopting a dynamic divide and conquer approach, where at every time instant each evader is assigned to a set of pursuers based on the instantaneous positions of all the players. The original MPME problem is decomposed to a sequence of simpler multiple pursuers single evader (MPSE) problems by testing whether a pursuer is relevant or redundant against each evader, by using Apollonius circles. Then, only the relevant pursuers participate in the MPSE pursuit of each evader. 
Recent surveys on pursuit evasion problems are \cite{chung2011search,kumkov2017zero,weintraub2020introduction}.

In \cite{chung2011search}, a taxonomy of search problems is presented. The paper highlights algorithms and results arising from different assumptions on searchers, evaders and environments and discusses potential field applications for these approaches. The authors focus on a number of pursuit-evasion games that are directly connected to robotics and not on differential games which are the focus of the other cited surveys. The paper concentrates on adversarial pursuit-evasion games on graphs and in polygonal environments where the objective is to maximize the worst-case performance on the search or capture time and on probabilistic search scenarios where the objective is the optimization of the expected value of the search objective, such as the maximal probability of detection or minimal capture time. In \cite{kumkov2017zero}, a survey on pursuit problems with $1$ pursuer versus $2$ evaders or $2$ pursuers versus $1$ evader are formulated as a dynamic game and solved with general methods of zero-sum differential games.
In \cite{weintraub2020introduction}, the authors present a recent survey on pursuit-evasion differential games and classify the papers according to the numbers of participating players: single-pursuer single-evader (SPSE), MPSE, one- pursuer multiple-evaders (SPME) and MPME.

In \cite{garcia2019optimal}, a two-player differential game in which a pursuer aims to capture an evader before it escapes a circular region is investigated. The state space, comprised of pursuer and evader locations, is divided into evader and defender winning regions. In each region the players try to execute their optimal strategies.
The players’ strategy depends on the state of the system (if it is in the capture or escape regions), and the proposed approach guarantees that if the players execute the prescribed optimal moves they improve their chances to win. The players move at a constant speed and the pursuer is faster than the evader. The players’ controls are the instantaneous heading angles. The game is a two-termination set differential game, i.e., the game ends when either player wins. In \cite{garcia2020multiple}, the problem of a border defense differential game where $M$ pursuers cooperate in order to optimally catch $N$ evaders before they reach the border of the region and escape is investigated. The members of the pursuer team exchange information between the team members and decide on the discrete assignment of pursers to evaders in an on-line manner. Furthermore, the game is a perfect information differential game where both pursuers and evaders have access to all state variables, which are the locations of all players, as well all their dynamics and velocities. The pursuers in this setting are assumed to have greater speeds than the evaders. The game takes place in a simple half-plane environment, and ends when all evaders are either caught or reach the border and escape.

Devising multi-robot perimeter patrol policies for adversarial settings in which an opponent has complete knowledge of the robots' patrol strategy are developed in \cite{agmon2008multi,agmon2011multi}. Possessing information about the robots' patrol policy enables the smart opponent to attempt to enter the perimeter undetected at the location with the highest probability of success. In order to prevent the opponent to utilize its knowledge on the strategy of the robot team, randomness is introduced into the robots' perimeter patrol algorithm, thus preventing the opponent from having full knowledge of the chosen patrol strategy and consequently increasing the chances to detect it.

\subsection{Contributions}

In this paper, we provide several theoretical and experimental contributions to multi-agent search and coordinated motion planning literature. In the considered scenario, a defending team of agents has to protect an initial region from the entrance of an unknown number of smart invaders, that have superior sensing and planning capabilities compared to the defender team. An analysis on the defenders' trajectories and critical speeds that enable the successful completion of the defense task is provided. Additionally, if possible, an additional goal for the defender team is to optimally expand the region which they guard to the maximal allowable size that still enables the defender team to detect any number of smart invaders that may attempt to enter the protected region. Hence, the total search times and the maximal attainable protected area are also reported.  Extensive theoretical and numerical analysis is performed for both the defense and maximal expansion tasks.
\begin{itemize}
        \item We propose two types of novel guarding and expansion strategies for any for even number of defenders:
        \begin{itemize}
        \item  Circular defense pincer sweep strategy 
        \item  Spiral defense pincer sweep strategy
        \end{itemize}
        \item Based on geometric and dynamic constraints we establish the necessary critical speed for each defense protocol to be successful and derive analytical expressions for the search times and radius of maximal expansion for the two types of search patterns. 
        \item We show that the spiral defense pincer sweep strategy enables defenders that sweep with speeds that are only slightly above the theoretical lower bound to ensure all invaders attempting to enter the protected region are detected. 
        \item We compare between the circular and spiral defense pincer sweep expansion strategies and highlight the advantages of the spiral strategy in both enabling the expansion of the region to a larger size and the search time required to reach it.
        \item We provide a quantitative comparison between the developed pincer-based defense protocols and the corresponding same-direction sweep protocols that are regarded as the state-of-the-art in defense against smart invaders. We prove that the corresponding pincer-based protocols yield lower critical speeds, shorter time to increase the protected region to its maximal size as well as the ability to expand the protected region to a larger area compared to same-direction protocols.
        \item We demonstrate through empirical simulations conducted with MATLAB and NetLogo the theoretical results and present the evolution of the guarding and expansion strategies graphically.
\end{itemize}

\subsection{Paper Organization}
This article is organized as follows. Section \ref{sec:universal_bound} presents an optimal bound on the speed of defenders employing the guarding task. This lower bound is independent of the actual implemented defense protocol. The obtained lower bound is used as a one of the comparison metrics for evaluation the performance of different guard/defense protocols. Section \ref{sec:circular_defense} provides an analysis of the defense and maximal expansion problems when the defender team performs the circular defense pincer sweep process. Section \ref{sec:spiral_defense} presents an analysis of the defense and maximal expansion problems when the defender team performs the spiral defense pincer sweep process. Section \ref{sec:comparative_analysis} provides a comparison between the circular and spiral defense pincer defense strategies and highlights the advantages of using the proposed spiral defense pincer protocol. Section \ref{sec:comparison_to_same_direction_protocols} compares prevalent approaches for defense against smart invaders which are considered the state-of-the-art in defense against smart invaders and compares these approaches to the pincer-based defense protocols developed in this work, proving the superiority of pincer-based approaches. The last section draws conclusions from the performed analysis and provides some interesting future research directions.

\section{Pincer-Based Defense}
\label{sec:pincer_based_defense}
This research focuses on developing a guaranteed defense protocol of an initial region from the entrance of an unknown number of smart invaders. The region is protected by employing a multi-agent team of identical cooperating defenders that sweep around the protected region and detect invaders that attempt to enter it. The defenders possess a linear sensor of length $2r$ with which they detect invaders that intersect their field-of-view. The only information the defenders have is that invaders may be located at any point outside of an initial circular region of radius $R_0$, referred to as the initial protected region at the beginning of the defense process. 

The proposed defense strategies are deterministic and pre-planned, and therefore they can be accomplished by using simple agents-like defenders. All defenders move with a speed is $V_s$ which is measured at the center of a defender's sensor. The invaders move at a maximal speed of $V_T$, and do not have any maneuverability restrictions. There are two objectives for a defense strategy, defending the initial protected region and, if possible, expanding the protected region by performing an iterative expansion strategy until the region reaches the maximal defendable area.

The time it takes the defender team to expand the protected region to the maximal defendable size obviously depends on the applied defense protocol. Two types of defense strategies are investigated, circular and spiral. When defenders perform the maximal expansion task, their goal is to iteratively increase the radius of the protected region after each sweep, up to the maximal defendable size of the region. At the beginning of the circular defense pincer sweep process only half the length of the defenders' sensors is outside of the protected region, i.e. a footprint of length $r$, while the other half is inside the region in order to catch invaders that may move inside the region while the search progresses. At the beginning of the spiral defense pincer sweep process the entire length of the defenders' sensors is outside the protected region, i.e. a footprint of length $2r$. 

The basic idea in performing pincer-based defense is requiring defenders to search for invaders in opposite directions instead of equally distributing the defenders around the region and letting them search in the same direction.  Defenders performing pincer movements address the worst-case scenario of invader entrance to the protected region from the ``most dangerous points", points situated at the edges of their sensors nearest to the protected region's boundary. Invaders located at these points have the maximal amount of time to advance towards the protected region during defender movement and thus if invaders that attempt to infiltrate the region from these points are detected, invaders trying to enter from other points are detected as well. For an extensive discussion about the comparison between pincer-based sweeps and same-direction sweeps in search tasks against smart evaders see \cite{francos2021swarms,francos2021pincer}.

Successfully completing the defense and maximal expansion tasks with the lowest possible critical speed is one of the performance criteria for an efficient defense strategy. Pincer-based defense procedures result in lower critical speeds compared to their same-direction counterparts, and hence are chosen in the developed defense protocols. The discussed pincer-based strategies can be performed with any even number of defenders. The basic idea of pincer-based defense is to decompose the defender team into pairs that are placed back-to-back at the start of the protocol. Each defender in a pair moves in an opposite direction, counter-clockwise or clockwise. When two defenders meet at a location after the completion of a sweep, implying that their sensors are again back-to-back, they switch their movement direction.  The direction switches are performed every time a defender meets the defender scanning the adjacent angular section to its section. Based on the numbers of defenders performing the defense task, the protected region is portioned into equal angular sectors, where each sector is searched by a different defender. The discussed defense protocols can be applied to both $2$ dimensional defense tasks on the plane or in $3$ dimensional defense tasks where the defenders and invaders are drone-like agents that fly over the protected region.

\begin{figure}[!htb]
\noindent \centering{}\includegraphics[width=1.55in,height =3.3in]{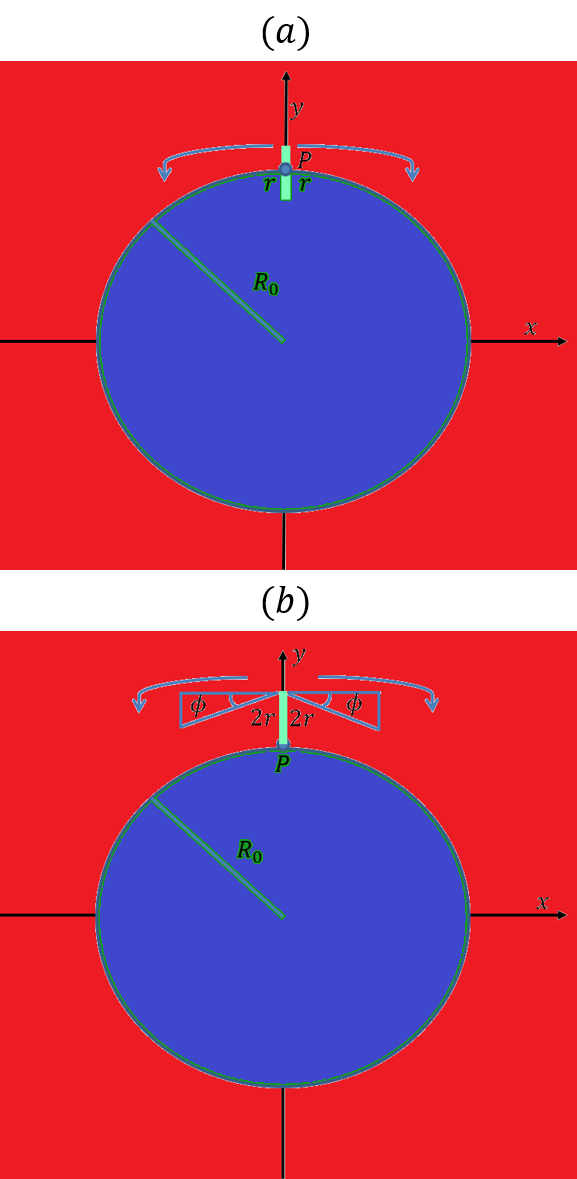} \caption{(a) - Initial placement of $2$ defenders performing the circular defense pincer sweep process.  (b) - Initial placement of $2$ defenders performing the spiral defense pincer sweep process. Defenders' sensors are shown in green and red areas indicate locations where potential invaders may be present. Blue areas represent locations that belong to the initial protected region that does not contain invaders. The angle $\phi$ is the angle between the tip of a defender's sensor and the normal to of the edge of the protected region. $\phi$ is an angle that depends on the ratio between the defender and invader speeds.}
\label{Fig1Label}
\end{figure}

The first considered defense protocol is the circular defense sweep protocol. The circular defense protocol allows to perform the defense and maximal expansion tasks with defenders with basic motion capabilities however as a byproduct of its simplicity it is far from being optimal. Hence, we propose the spiral defense pincer sweep process that provides an improved protocol that uses spiral scans, drawing inspiration from \cite{mcgee2006guaranteed}. The spiraling-in trajectories of the defenders allow them to track the "wavefront" of the expanding protected region, thus detecting invaders at the furthest possible locations from the invader region. At last, we compare and discuss the obtained results of the two defense strategies. The evaluation metrics for the defense strategies include the minimal defender speed required for successful defense of the initial protected region, the time to expand the protected region to the maximal defendable area as well as the maximal feasible protected region's radius resulting from the defense protocol. All these quantities are expressed as a function of the search parameters $R_0$, $r$, $V_T$ and the number of defenders, $n$.

Illustrative simulations demonstrating the evolution of the defense processes were generated using NetLogo software \cite{tisue2004netlogo} and are presented in Fig. $2$. and Fig. $3$. Green areas show locations that were searched and hence do not contain invaders and red areas indicate locations where potential invaders may be present. Blue areas represent locations that belong to the initial protected region that does not contain invaders. Fig. $2$. shows the cleaning progress during the expansion of the protected region when $6$ defenders perform the circular defense pincer sweep process. Fig. $3$. shows the evolution of the defense process during the expansion of the protected region when $4$ defenders perform the spiral defense pincer sweep process.

\begin{figure}[!htb]
\noindent \centering{}\includegraphics[width=3.5in,height =3.80418in]{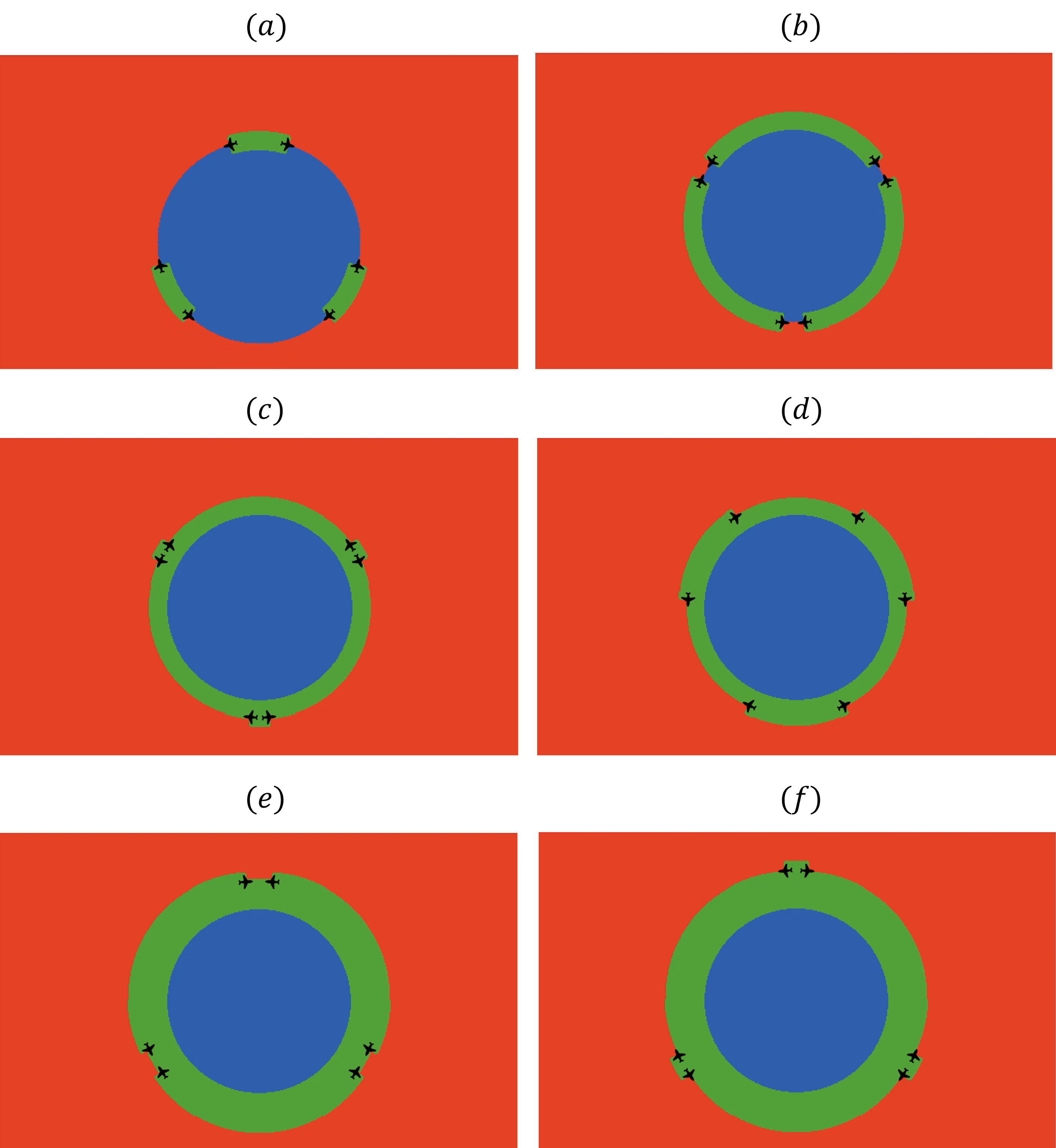} \caption{Swept areas and protected region status for different times in a scenario where $6$ defenders perform the circular defense pincer sweep process. (a) - Beginning of first sweep. (b) - Toward the completion of the first sweep. (c) - Beginning of the second sweep. (d) - Midway of the second sweep. (e) - End of the fourth sweep. (f) - Beginning of fifth sweep. Green areas show locations that were searched and hence do not contain invaders and red areas indicate locations where potential invaders may be present. Blue areas represent locations that belong to the initial protected region that does not contain invaders.}
\label{Fig2Label}
\end{figure}

\begin{figure}[!htb]
\noindent \centering{}\includegraphics[width=3.5in,height =3.85875in]{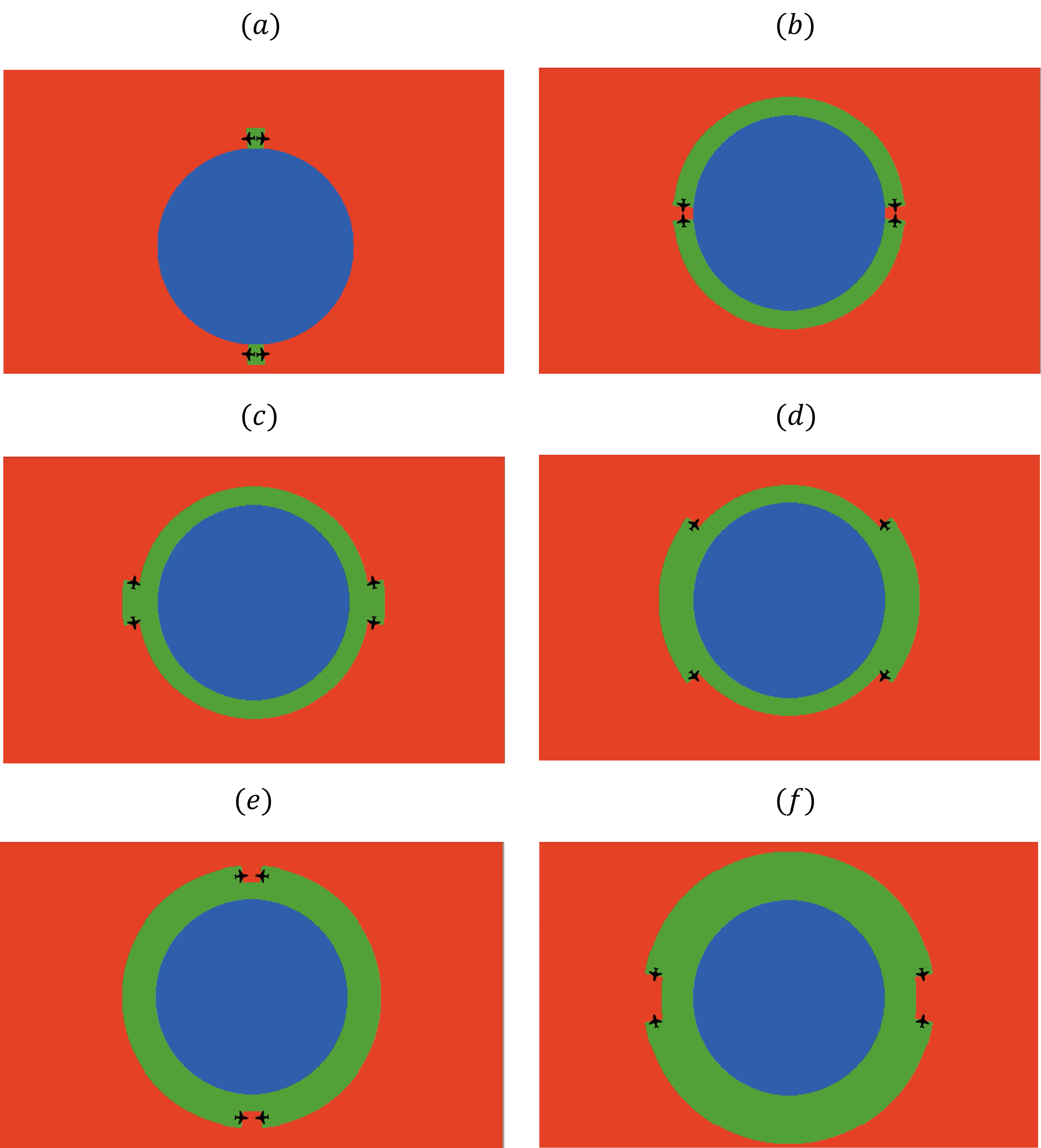} \caption{Swept areas and protected region status for different times in a scenario where $4$ defenders perform the spiral defense pincer sweep process. (a) - Beginning of first sweep. (b) - End of the first sweep. (c) - Beginning of the second sweep. (d) - Midway of the second sweep. (e) - End of the second sweep. (f) - Toward the end of the third sweep. Green areas show locations that were searched and hence do not contain invaders and red areas indicate locations where potential invaders may be present. Blue areas represent locations that belong to the initial protected region that does not contain invaders.}
\label{Fig3Label}
\end{figure}

Note that in the considered problems, the search is continued until the expansion of the protected region reaches the maximal attainable radius, and afterwards the defenders continuously patrol around this radius.   

\subsection{Comparison to Related Research}

In our previous work \cite{francos2019search}, the confinement of an unknown number of smart evaders that are originally located somewhere inside a given circular region to their original domain is investigated. By deploying a line formation of searching agents or alternatively a single agent with an equivalent sized linear sensor that sweep inside and around the region, guaranteed detection protocols are developed. In case the speed of the agents in the line formation exceeds a critical speed, they may decrease the evader region by performing a search protocol consisting of alternating circular sweeps around the region that are followed by inward advancement steps toward the center of the evader region ( the region in which evaders are located). A proof in the paper shows that since the evaders are smart, a search pattern that uses circular sweeps cannot completely clean the evader region, Therefore, after the evader region is reduced by the circular sweeping protocol to a region that is bounded in a circle with a radius equal to half the formation's sensing range, the search pattern must be changed in order to perform a set of end-game maneuvers that allow to guarantee the detection of all smart evaders in the region. 

In \cite{francos2021swarms}, we consider teams of agents that perform pincer sweep search strategies with linear sensors, in order to detect all smart evaders that try to escape from a given region. The paper presents two multi-agent pincer sweep search strategies, circular and spiral, that can be applied with any even number of sweeping agents. The results obtained from the paper show that performing the circular pincer sweep process, where pairs of defenders sweep toward each other allow for lower critical speeds and shorter sweep times until the entire evader region is searched and cleared from evaders compared to a circular sweep process in which the defenders are equally distributed around the evader region and all sweep in the same direction. A circular sweep process in which the searchers all rotate in the same direction is the extension of \cite{francos2019search} to a scenario where the defenders are distributed equally around the region and perform the circular sweep protocol described in \cite{francos2019search}. Therefore, defenders that rotate in the same direction have to perform an end-game scenario similar to the one described in \cite{francos2019search} in order to completely clear the region from evaders, a set of maneuvers that is unnecessary when using the circular pincer sweep protocol described in \cite{francos2021swarms}. 

The second type of sweep pattern that is developed in \cite{francos2021swarms} is the spiral pincer sweep process that allows to complete the search of the evader region in a significantly shorter time and at lower critical speeds compared to the circular pincer sweep process. The critical speed of the spiral pincer sweep process approaches the theoretical lower bound on the critical speed.  

This work aims to solve the dual problem to the problem investigated in \cite{francos2019search} and \cite{francos2021swarms}, which is to protect a given initial region that does not contain invaders from their entrance, and if possible to expand the protected region to the maximal defendable size possible. The first task this work is concerned with, the guarding task is analogous to the confinement task in \cite{francos2019search,francos2021swarms}, in the sense that it aims to keep the protected region's radius constant after the defenders complete a full sweep around the region. The maximal expansion task, in which after each full sweep around the protected region, its radius increases is the dual problem to the constriction of the evader region in \cite{francos2019search,francos2021swarms}. 

In contrast to the algorithms described in \cite{francos2019search,francos2021swarms} the defense process does not terminate when the defenders expand the protected region to its maximal size and they must continuously sweep around the region to keep the intruders that are outside of the protected region from entering it. Additionally, this work presents for each developed pincer sweep defense process the maximal radius that the protected region can be extended into and presents an analysis on the trade-off between approaching the maximal protected region's radius and the sweep time it takes to reach it. Alternatively to the barrier placement problem in \cite{mcgee2006guaranteed}, our approach emphasizes the usage of cooperation between the defenders by using pincer sweeps, calculates the maximal defendable region's size and presents analytical solutions to all aspects of the defense and maximal expansion problems against smart invaders.

The research conducted in \cite{shishika2018local,shishika2019team,shishika2020cooperative} also investigates protecting a region from the entrance of invaders and uses pincer movements between pairs of defenders as well. However, the goal of the defender team in these works is to intercept the largest possible number of invaders contrary to our approach which develops a defense protocol that ensures detection of all invaders, regardless of their numbers, and consequently sets necessary requirements on the defender team in order to achieve its goal. Furthermore, in \cite{shishika2018local,shishika2019team,shishika2020cooperative}, the invaders' locations are constantly known to the defenders hence this information assists the defender team in planning and coordinating its movements and its allocation of defender pairs. Conversely, our approach does not assume any knowledge on the invaders' locations or their number. 

\section{A Universal Bound On Cleaning Rate}
\label{sec:universal_bound}
\noindent In this section we present an optimal bound on the cleaning rate of a defender with a linear shaped sensor. This bound is independent of the particular defense pattern employed. For each of the proposed defense methods we then compare the resulting cleaning rate to the optimal derived bound in order to compare between different defense protocols. We denote the defender's speed as $V_s$, the sensor length as $2r$, the protected region's initial radius as $R_0$ and the maximal speed of an invading agent as $V_T$. The maximal cleaning rate occurs when the footprint of the sensor outside the protected region is maximal. For a line shaped sensor of length $2r$ this happens when the entire length of the sensor is fully outside the protected region and it moves perpendicular to its orientation. The rate of sweeping when this happens has to be higher than the minimal decrease rate of the protected region (given its total area) otherwise no sweeping process can ensure detection of all invaders. We analyze the defense process when the defender swarm comprises $n$ identical agents. The smallest defender's speed satisfying this requirement is defined as the critical speed and denoted by $V_{LB}$, we have:
\newtheorem{thm}{Theorem}
\begin{thm}
No sweeping process is able to successfully complete the defense task if its speed, $V_s$, is less than,
\begin{equation}
{{\rm{V}}_{LB}} = \frac{{\pi {R_0}{V_T}}}{nr}
\label{e1}
\end{equation}
\end{thm}

For proof see \cite{francos2019search}. The desired outcome it that after the first sweep the protected region is within a circle with a greater radius than the initial protected region's radius. 

\section{Circular Defense Pincer Sweep Process}
\label{sec:circular_defense}

At first, we study a scenario in which a multi-agent team comprising of $n$ agents, referred to as defenders, perform the defense task. The number of defenders, $n$, is even, and all defenders are identical and are equipped with a linear sensor with a length of $2r$. The initial defenders' configuration at the start of the defense protocol is such that each defender has half of its linear sensor outside of the protected region, i.e. a length of $r$. 

\begin{figure}[!htb]
\noindent \centering{}\includegraphics[width=3.5in,height =2.4663in]{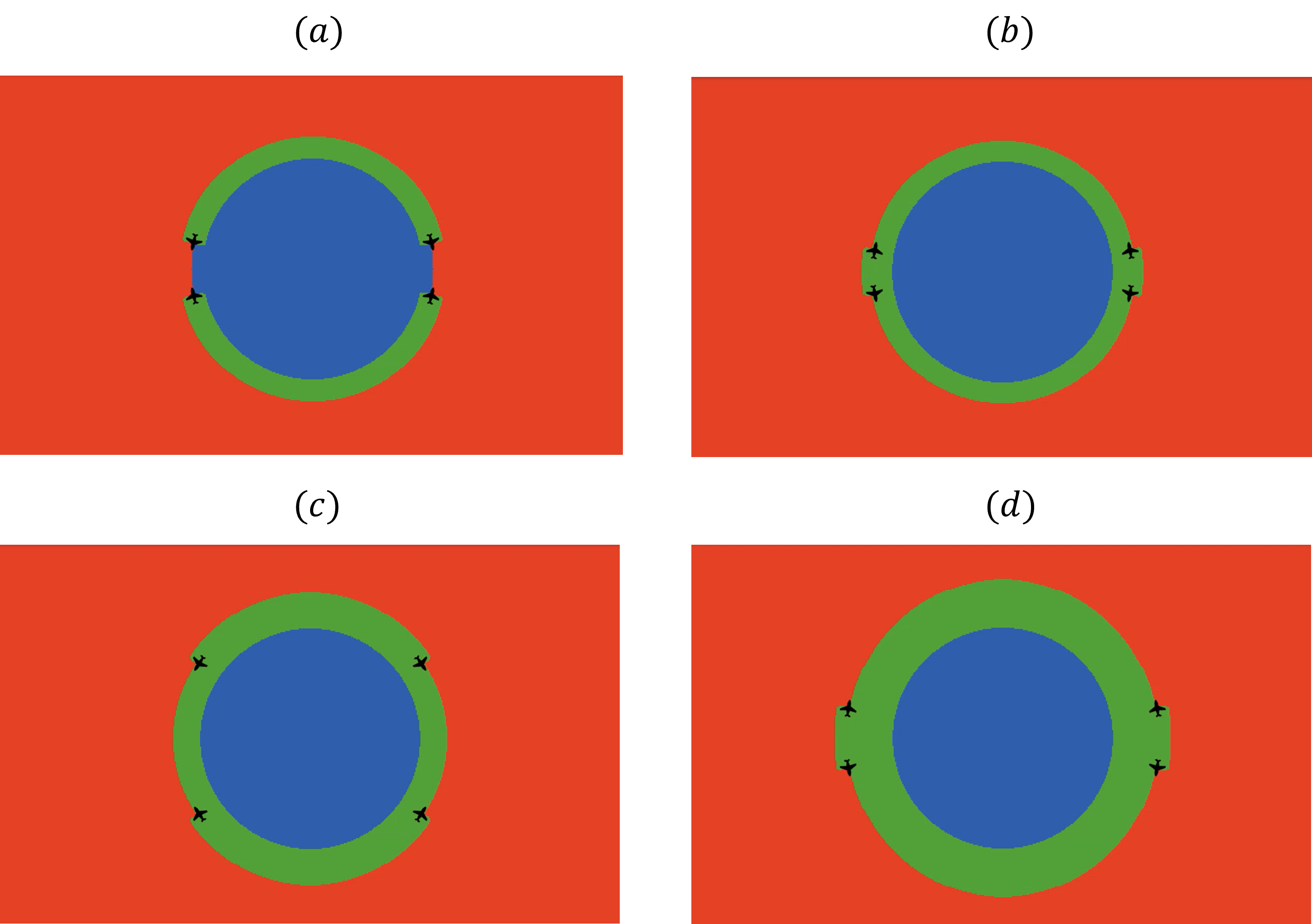} \caption{Swept areas and protected region status for different times in a scenario where $4$ defenders perform the circular defense pincer sweep process. (a) - Beginning of first sweep. (b) - Toward the end of the first sweep. (c) - Beginning of second sweep. (d) - Midway of the third sweep. Green areas show locations that were searched and hence do not contain invaders and red areas indicate locations where potential invaders may be present. Blue areas represent locations that belong to the initial protected region that does not contain invaders.}
\label{Fig4Label}
\end{figure}

\subsection{The Defense Task and Critical Speed Analysis}
Using pincer movements between adjacent pairs of defenders leverages the symmetry between the trajectories of nearby defenders in order to impede the entrance of invaders to the protected region from the most dangerous points invaders can enter from (by using similar arguments to the proof provided in (\cite{francos2019search}). Hence, the defenders' critical speed is computed based on the time required for a defender to scan the sector allocated for it, i.e. an angular section of $\frac{2\pi}{n}$. In case the defenders' speeds exceed the critical speed required for successfully implementing the defense/guarding task, the defenders can advance outwards together from the center of the protected region after completing a sweep. A full sweep or iteration refers to a defender's scan of the sector of the protected region it guards spanning an angle of $\frac{2\pi}{n}$. Hence, the scanned angle each defender guards is a function of the participating defenders in the defense task.  When defenders perform the defense task, they change their scanning direction after the completion of a full sweep. 
\begin{thm}
The circular critical speed equals twice the optimal minimal critical speed,
\begin{equation}
{V_c} = 2V_{LB}
\label{e45}
\end{equation}
\end{thm}
\begin{proof}
Every defender performing the circular defense pincer sweep process has a sensor length of $r$ inside the protected region, to ensure no invader enters the protected region without being detected by the defenders. Hence, in order to catch all invaders, the spread from any potential location inside the invader region (the region outside of the protected region where invaders might be located), has to be restricted to a radius less than $r$ from its origin point at the start of the defense protocol. Therefore, during an angular traversal of $\frac{2\pi}{n}$ around a protected region with a radius of $R_0$, this requirement implies that,
\begin{equation}
\frac{{2\pi {R_0}}}{{n{V_s}}} \le \frac{r}{{{V_T}}}
\label{e22}
\end{equation}
Hence, in order to detect all invaders, the defenders' speed has to exceed,
\begin{equation}
{V_s} \ge \frac{{2\pi {R_0}{V_T}}}{{nr}}
\label{e23}
\end{equation}
The critical speed for the circular defense pincer sweep process is obtained when (\ref{e23}) is satisfied with equality.

\begin{equation}
{V_c} = \frac{{2\pi {R_0}{V_T}}}{{nr}}
\label{e24}
\end{equation}
\end{proof}

The obtained result matches the observation that the circular critical speed for defending a region against the entrance of smart invaders should equal the critical speed required for confining smart evaders (developed in \cite{francos2021swarms}) inside a region having the same size.  
\subsection{The Maximal Expansion Task}
\subsubsection{Number of Sweeps Analysis}
\begin{thm}
The maximal radius that $n$ circularly sweeping defenders, with a linear sensor of length $2r$, a given speed $V_s$ and a maximal invader speed of $V_T$ can safely protect against the entrance of invaders is,
\begin{equation}
{{\bar R}_{N_c}} = \frac{{n{V_s}r}}{{2\pi {V_T}}}   
\label{e1096}
\end{equation}
\end{thm}

\begin{thm}
For a defender team with $n$ defenders, for which $n$ is even, performing the circular defense pincer sweep process, the number of sweeps required for the defender team to expand the protected region to be bounded by a circle with a radius that is $\varepsilon$ close to the maximal boundable radius ${{\bar R}_{N_c}}$ is,
\begin{equation}
{N_n} = \left\lceil {\frac{{\ln \left( {\frac{{ - 2\pi {V_T}\varepsilon }}{{2\pi {R_0}{V_T} - n{V_s}r}}} \right)}}{{\ln \left( {1 - \frac{{2\pi {V_T}}}{{n\left( {{V_s} + {V_T}} \right)}}} \right)}}} \right\rceil 
\label{e98986}
\end{equation}
Denote by ${T_{out}}$ the sum of all outward advancement times and by ${T_{circular}}$ the sum of all circular search times. Hence, the total search time necessary for the defender team to expand the protected region to its maximal size is given by,
\begin{equation}
T(n) = {T_{out}}(n) + {T_{circular}}(n)
\label{e98988}
\end{equation}
${T_{out}}(n)$ is given by,
\begin{equation}
{T_{out}}(n) = \frac{{nr}}{{2\pi {V_T}}} - \frac{{{R_0} + \varepsilon }}{{{V_s}}}
\label{e98989}
\end{equation}
And ${T_{circular}}(n)$ is given by,
\begin{equation}
\begin{array}{l}
{T_{circular}}(n) = \frac{{{R_0}\left( {{V_s} + {V_T}} \right)}}{{{V_s}{V_T}}} - \frac{{rn\left( {{V_s} + {V_T}} \right)}}{{2\pi {V_T}^2}} \\ - {\left( {1 - \frac{{2\pi {V_T}}}{{n\left( {{V_s} + {V_T}} \right)}}} \right)^{{N_n}}}\left( {\frac{{2\pi {R_0}{V_T} - rn{V_s}}}{{{V_s}{V_T}}}} \right)\left( {\frac{{{V_s} + {V_T}}}{{2\pi {V_T}}}} \right) + \frac{{{N_n}r}}{{{V_T}}}
\end{array}
\label{e98990}
\end{equation}
\end{thm}
\begin{proof}
Denote by $\Delta V >0$ the excess speed of the defender above the critical speed. Hence, the defender's speed is  $V_s = V_c + \Delta V$. The time required for a defender to circularly sweep around its allocated section is,
\begin{equation}
{T_{circular}}_i = \frac{{2\pi {R_i}}}{{n({V_c} + \Delta V)}}
\label{e25}
\end{equation}
Since $V_s = V_c + \Delta V$, ${T_{circular}}_i$ may also be expressed as,
\begin{equation}
{T_{circular}}_i = \frac{{2\pi {R_i}}}{{n{V_s}}}
\label{e20}
\end{equation}
Depending on the number of participating defenders and the iteration number, the distance a defender can advance outwards from the center of the protected region after completing an iteration is,
\begin{equation}
{\delta _i}(\Delta V) = r - {V_T}{T_{circular}}_i \hspace{1mm}, \hspace{1mm} 0  \le {\delta _i}(\Delta V) \le r
\label{e1090}
\end{equation}
In the expression ${\delta _i}(\Delta V)$, $\Delta V$ is the excess speed of a defender with respect to the critical speed. Denote by $i$ the number of full sweeps defenders perform around the protected region, in which the first sweep occurs when $i=0$. After completing a full sweep, the defenders move outwards from the center of the protected region with the inner tips of their sensors pointing to the center of the protected region. At times in which defenders move outwards, they do so with a speed of $V_s$. This motion continues up to the location in which the defenders begin their next sweep once half of their sensor is inside the expanding wavefront of the invading region. During the outward advancements no invaders are detected, while the protected region continues to shrink.

The time it takes defenders to move outwards up to the point where half of their sensors are outside of the protected region depends on the relative speed between the defenders outward motion and the invader region's inwards expansion and is given by (\ref{e1068}). As the defenders gradually head outward from the center of the protected region, the protected region continuous to shrink. Hence, in order for no invader to enter the region, defenders must advance outwards to a lesser extent than ${\delta _i}(\Delta V)$. This quantity is denoted by ${\delta _{{i_{eff}}}}(\Delta V)$, and depends on the ratio between the speed in which defenders move outwards from the center of the protected region and the sum of the defender and invader region spread speeds. ${\delta _{{i_{eff}}}}(\Delta V)$ is the actual distance defenders are allowed to progress outwards after each sweep so that they meet the wavefront of the expanding invader region at the point where half of their sensors are over the invader region. Therefore, the distance a defender may advance outwards after completing a sweep around the protected region is,
\begin{equation}
{\delta _{{i_{eff}}}}(\Delta V) = {\delta _i}(\Delta V)\left( {\frac{{{V_s}}}{{{V_s} + {V_T}}}} \right)
\label{e500}
\end{equation}
The outward advancement time depends on the iteration number. It is denoted by $T_{out{_i}}$ and is expressed as,
\begin{equation}
T_{{out}_i} = \frac{{{\delta _{{i_{eff}}}}(\Delta V)}}{{{V_s}}} = \frac{{rn{V_s} - 2\pi {R_i}{V_T}}}{{n{V_s}\left( {{V_s} + {V_T}} \right)}}
\label{e1068}
\end{equation}
The index $i$ in ${T_{out_i}}$ denotes the iteration number in which the advancement takes place. After the defenders complete their sweep, the protected region is bounded by a circle with a larger radius compared to the previous sweep. Thus, the new radius of the circle bounding the protected region is given by,
\begin{equation}
{R_{i + 1}} = {R_i} + {\delta _{{i_{eff}}}}(\Delta V) = {R_i} + {\delta _i}(\Delta V)\left( {\frac{{{V_s}}}{{{V_s} + {V_T}}}} \right)
\label{e1069}
\end{equation}
Replacing the value of ${\delta _i}(\Delta V)$ from (\ref{e1090}) into (\ref{e1069}) results in,
\begin{equation}
{R_{i + 1}} = {R_i} + {\delta _{{i_{eff}}}}(\Delta V) = {R_i} + \frac{{r{V_s}}}{{{V_s} + {V_T}}} - \frac{{2\pi {R_i}{V_T}}}{{n\left( {{V_s} + {V_T}} \right)}}
\label{e1093}
\end{equation}
Rearranging terms yields,
\begin{equation}
{R_{i + 1}} = {R_i}\left( {1 - \frac{{2\pi {V_T}}}{{n\left( {{V_s} + {V_T}} \right)}}} \right) + \frac{{r{V_s}}}{{{V_s} + {V_T}}}
\label{e1094}
\end{equation}
Denote the coefficients $c_1$ and $c_2$ by,
\begin{equation}
{c_1} =  \frac{{r{V_s}}}{{{V_s} + {V_T}}},{c_2} = 1 - \frac{{2\pi {V_T}}}{{n\left( {{V_s} + {V_T}} \right)}}
\label{e1095}
\end{equation}
Hence, (\ref{e1094}) can be expressed as,
\begin{equation}
{R_{i + 1}} = {c_2}{R_i} + {c_1}
\label{e1300}
\end{equation}

Since the defenders need to move outwards from protecting a region with a smaller radius, and during this outwards motion the protection region continues to shrink, the defenders are able to protect a marginally smaller region. The expansion task progresses by alternating circular sweeps and outward movements until the protected region is bounded by the largest possible circle. Let $\varepsilon>0$, and denote by ${\widehat{R}_{{N_n} - 1}}$ the radius of the protected region that is $\varepsilon$ close to ${{\bar R}_{N_c}}$,
\begin{equation}
\widehat{R}{_{{N_n} - 1}} = R_{max} = \frac{{n{V_s}r}}{{2\pi {V_T}}} - \varepsilon 
\label{e1091}
\end{equation}
The number of iterations required for the defender team to expand the protected region to be bounded by a circle with a radius of $\widehat{R}{_{{N_n} - 1}}$ that is $\varepsilon$ close to ${{\bar R}_{N_c}}$ is calculated by similar steps as the calculation in Appendix $A$ of \cite{francos2021swarms}. It is given by,
\begin{equation}
N_n = \left\lceil 
\frac{1}{\ln {c_2}}\ln \left( {\frac{{\widehat{R}{_{{N_n} - 1}} - \frac{{{c_1}}}{{1 - {c_2}}}}}{{{R_0} - \frac{{{c_1}}}{{1 - {c_2}}}}}} \right)
\right\rceil
\label{e11}
\end{equation}
The radius $R_{max}= \widehat{R}{_{{N_n} - 1}}$ is the maximal radius the protected region expands to, and is used to calculate the number of sweeps required to reach this radius. The actual radius that bounds the protected region after $N_n$ sweeps is denoted by ${R}{_{{N_n} - 1}}$ and is computed after $N_n$ is calculated. Replacing the coefficients in (\ref{e11}) yields that the number of sweeps required for the defenders to increase the protected region to be in a circle with the radius of the last sweep around the region, ${R_{{N_n} - 1}}$, is,
\begin{equation}
{N_n} = \left\lceil {\frac{{\ln \left( {\frac{{ - 2\pi {V_T}\varepsilon }}{{2\pi {R_0}{V_T} - n{V_s}r}}} \right)}}{{\ln \left( {1 - \frac{{2\pi {V_T}}}{{n\left( {{V_s} + {V_T}} \right)}}} \right)}}} \right\rceil   
\label{e1019}
\end{equation}
To determine the number of required sweeps, the ceiling operator is used in order to implicitly demand that the number of iterations is an integer number, thus causing defenders to complete their sweep cycle and meet the defender that searches the adjacent section. This leads defenders to finish sweep $N_n-1$, even when the protected region's radius is somewhat larger than ${\widehat{R}_{{N_n} - 1}}$. 

After completing the last circular sweep, defenders perform the last outward advancement, and defenders continue to circularly patrol around a protected region with a radius of $R_{max}$ indefinitely. The total search time required for a defender team of $n$ defenders to enlarge the protected region to its maximal size is obtained by combining the total outward advancement times together with the total circular sweep times around the protected region in all iterations. Denote by ${T_{out}}(n)$ the sum of all the outward advancement times and by ${T_{circular}}(n)$ the sum of all circular sweep times. Hence,
\begin{equation}
T(n) = {T_{out}}(n) + {T_{circular}}(n)
\label{e1101}
\end{equation}
\subsubsection{Outward Advancement Times Calculation}
Denote the total outward advancement times until the protected region is within a circle of radius ${R_{N - 1}}$ as $\widetilde{T}_{out}(n)$. This time is calculated by,
\begin{equation}
\widetilde{T}_{out}(n) = \sum\limits_{i = 0}^{{N_n} - 2} {{T_{ou{t_i}}}  }
\label{e1099}
\end{equation}
Throughout the outward advancements phases the defenders do not perform sweeping and detection of invaders and hence invaders are not detected until the defenders finish their outwards motion and resume the sweeping of the protected region. Following a defender's completion of the outwards progression phase, its sensor overlaps both the invader region and the protected region by $r$.

Substituting ${T_{ou{t_i}}}$ in (\ref{e1099}) yields that the accumulation of outward advancement times prior to the protected region being within a circle of radius ${R_{N - 1}}$ can be calculated as follows,
\begin{equation}
\widetilde{T}_{out}(n)= \sum\limits_{i = 0}^{{N_n} - 2} {{T_{ou{t_i}}} = } \frac{{\left( {{N_n} - 1} \right)r}}{{{V_s} + {V_T}}} - \frac{{2\pi {V_T}\sum\limits_{i = 0}^{{N_n} - 2} {{R_i}} }}{{n{V_s}\left( {{V_s} + {V_T}} \right)}}
\label{e1020}
\end{equation}
The first outward advancement takes place after the protected region is within a circle of radius $R_0$ and the final outward advancement occurs after iteration number ${N_n}-2$, causing the protected region to expand from being inside a circle of radius ${R_{{N_n} - 2}}$ to being within a circle of radius ${R_{{N_n} - 1}}$. Following this motion, the defender team circularly sweeps around the protected region with a radius of ${R_{{N_n} - 1}}$. ${R_{{N_n} - 1}}$ that is calculated using the recursive relation in (\ref{e1300}) and is given by,
\begin{equation}
{R_{{N_n} - 1}} = \frac{{{c_1}}}{{1 - {c_2}}} + {c_2}^{{N_n} - 1}\left( {{R_0} - \frac{{{c_1}}}{{1 - {c_2}}}} \right)
\label{e505}
\end{equation}
Substitution of coefficients in (\ref{e505}) yields,
\begin{equation}
\begin{array}{l}
{R_{{N_n} - 1}} = \frac{{n{V_s}r}}{{2\pi {V_T}}} + {\left( {1 - \frac{{2\pi {V_T}}}{{n\left( {{V_s} + {V_T}} \right)}}} \right)^{{N_n} - 1}}\left( {\frac{{2\pi {R_0}{V_T} - n{V_s}r}}{{2\pi {V_T}}}} \right)
\label{e508}
\end{array}
\end{equation}
The full derivation of $\widetilde{T}_{out}(n)$ is continued in Appendix $A$. Hence,
\begin{equation}
\begin{array}{l}
\widetilde{T}_{out}(n) = \sum\limits_{i = 0}^{{N_n} - 2} {{T_{ou{t_i}}} }= \\   - \frac{{{R_0}}}{{{V_s}}} + \frac{{nr}}{{2\pi {V_T}}} + {\left( {1 - \frac{{2\pi {V_T}}}{{n\left( {{V_s} + {V_T}} \right)}}} \right)^{{N_n} - 1}}\left( {\frac{{2\pi {R_0}{V_T} - n{V_s}r}}{{2\pi {V_T}{V_s}}}} \right)
\end{array}
\label{e1100}
\end{equation}
Following the last circular sweep the defenders perform the last outward advancement, until reaching $R_{max}$. The time it takes the defenders to perform this last outward advancement is given by,
\begin{equation}
T_{out_{last}} = \frac{\left(\widehat{R}{_{{N_n} - 1}} - {R_{{N_n} - 1}}\right)}{V_s}
\label{e1102}
\end{equation}
After this last outward sweep, the defenders perform circular sweeps around a protected region of radius $R_{max}$ and continuously protect the region from the entrance of invaders after reaching the maximal protected region the defenders can guard.
Summing $\widetilde{T}_{out}(n)$ and the last outward advancement time in (\ref{e1102}) yields,
\begin{equation}
{T_{out}}(n) = \frac{{nr}}{{2\pi {V_T}}} - \frac{{{R_0} + \varepsilon }}{{{V_s}}}
\label{e1105}
\end{equation}

\subsubsection{Circular Sweep Times Calculation}
The time to complete the first circular sweep is ${T_0} = \frac{{2\pi {R_0}}}{{n{V_s}}}$. Similarly, the time it takes to perform a circular motion spanning an angular section of $\frac{2\pi}{n}$ around a circle of radius $R_i$ while moving with a speed of $V_s$ is,
\begin{equation}
{T_i} = \frac{{2\pi {R_i}}}{{n{V_s}}}
\label{e32}
\end{equation}
Denote the coefficient $c_3$ by,
\begin{equation}
{c_3} =  \frac{{2\pi r}}{{n\left( {{V_s} + {V_T}} \right)}}
\label{e30}
\end{equation}
It can be noted that by multiplying (\ref{e1300}) by $\frac{2\pi}{n{V_s}}$ one obtains a recursive difference equation that cab be utilized to calculate the circular sweep times. Therefore the sweep times can be written as,
\begin{equation}
T_{i + 1} = c_2{T_i} + c_3
\label{e12}
\end{equation}
Denote by $T_{circular}(n)$ the total circular sweep times required to expand the protected region to be within a circle having a radius equal to or greater than ${\widehat{R}_{N - 1}}$. The calculation of  $T_{circular}(n)$ follows similar steps as in appendix $C$ of \cite{francos2021swarms}. Hence,
\begin{equation}
T_{circular}(n)  = \frac{{{T_0} - {c_2}{T_{N_n - 1}} + \left( {{N_n} - 1} \right){c_3}}}{{1 - {c_2}}}
\label{e8}
\end{equation}
The calculation of the last circular sweep time prior to the protected region being within a circle having a radius greater than or equal to ${\widehat{R}_{N - 1}}$ follows similar steps as in appendix $D$ of \cite{francos2021swarms}. Hence,
\begin{equation}
{T_{{N_n} - 1}} = \frac{{{c_3}}}{{1 - {c_2}}} +  {{c_2}^{{N_n} - 1}}\left( {{T_0} - \frac{{{c_3}}}{{1 - {c_2}}}} \right)
\label{e10}
\end{equation}
Substitution of coefficients in (\ref{e10}) results in,
\begin{equation}
\begin{array}{l}
{T_{N - 1}} = \frac{r}{{{V_T}}} + {\left( {1 - \frac{{2\pi {V_T}}}{{n\left( {{V_s} + {V_T}} \right)}}} \right)^{{N_n} - 1}}\left( {\frac{{2\pi {R_0}{V_T} - rn{V_s}}}{{n{V_s}{V_T}}}} \right)
\end{array}
\label{e35}
\end{equation}
Therefore, the total circular sweep times from (\ref{e8}) are, 
\begin{equation}
\begin{array}{l}
T_{circular}(n)  =  \frac{{{R_0}\left( {{V_s} + {V_T}} \right)}}{{{V_s}{V_T}}} - \frac{{rn\left( {{V_s} + {V_T}} \right)}}{{2\pi {V_T}^2}} \\ - {\left( {1 - \frac{{2\pi {V_T}}}{{n\left( {{V_s} + {V_T}} \right)}}} \right)^{{N_n}}}\left( {\frac{{2\pi {R_0}{V_T} - rn{V_s}}}{{{V_s}{V_T}}}} \right)\left( {\frac{{{V_s} + {V_T}}}{{2\pi {V_T}}}} \right) + \frac{{{N_n}r}}{{{V_T}}}
\end{array}
\label{e1035}
\end{equation}
Subsequently to the completion of sweep $N_n$ the protected region is within a circle of ${R_{{N_n} - 1}}$. 
\end{proof}
\subsubsection{Numerical Experiments}

Fig. $5$ presents the maximal protected region's radius that the defenders are able to protect. The maximal radius clearly depends on the number of defenders and their speed. Fig. $6$ presents the number of sweeps required to expand the protected region to its maximal size as a function of $\varepsilon$. Fig. $7$ presents the expansion time of the protected region to $R_{N_c} - \varepsilon$ for a fixed speed exceeding the circular critical speed. Fig. $8$ presents the total search times for different numbers of defenders until maximal expansion of the protected region is achieved. In all presented graphs the defenders' speed is equal and is independent of the number of defenders performing the expansion protocols, and is chosen so that search times of defender teams with different number of defenders are correctly compared. The values of $\Delta V$ mentioned in the plots are speeds exceeding the critical speed of two defenders employing the circular defense pincer sweep process, due to the fact that as the number of defenders increases, the critical speed required for successfully performing the defense task decreases. Hence, defender teams with more defenders can achieve their goal of defending the region while moving at speeds exceeding the critical speed of two defenders, while the contrary argument is false. The second plot from the top of Fig. $8$. presents the search time reduction obtained when the number of participating defenders increases. 

\begin{figure}[!htb]
\noindent \centering{}\includegraphics{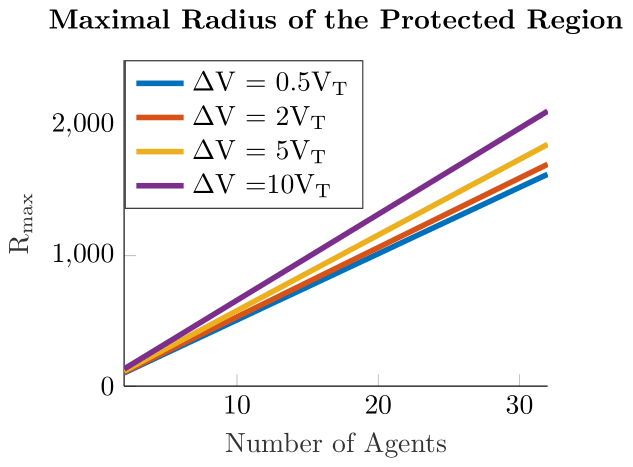} \caption{Maximal protected region's radius. We simulated circular defense sweep processes with an even number of agents, ranging from $2$ to $32$ agents. The chosen values of the parameters are $r=10$, $V_T = 1$ and $\varepsilon = 0.2$.}
\label{Fig5Label}
\end{figure}

\begin{figure}[!htb]
\noindent \centering{}\includegraphics{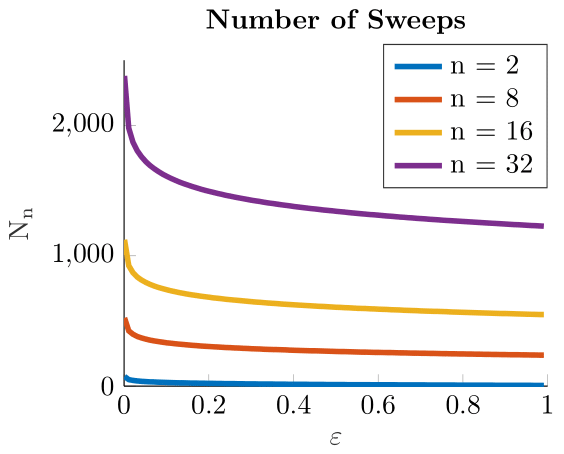} \caption{Number of sweeps required to expand the protected region to its maximal size as a function of $\varepsilon$. We plot the results for defenders performing the circular defense pincer sweep processes with $2,8,16$ and $32$ agents. The chosen values of the parameters are $r=10$, $V_T = 1$, $V_s = 31.9159$.}
\label{Fig6Label}
\end{figure}

\begin{figure}[!htb]
\noindent \centering{}\includegraphics{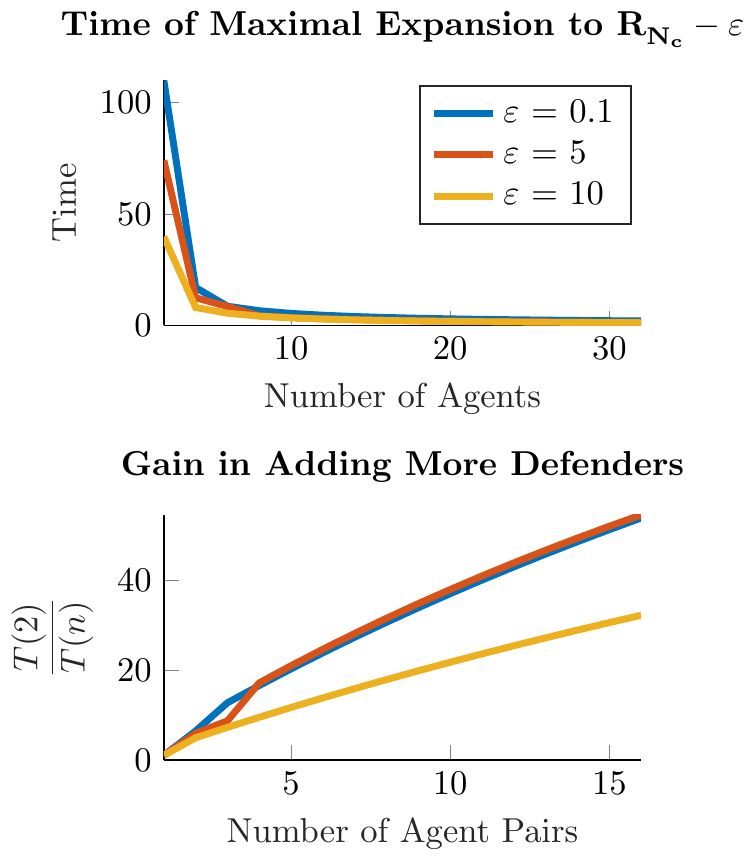} \caption{Time of maximal expansion of the protected region to $R_{N_c} - \varepsilon$ and gain in adding more defenders for equal defender speeds. We simulated the circular defense pincer sweep processes for an even number of agents, ranging from $2$ to $32$ agents. The chosen values of the parameters are $r=10$, $V_T = 1$, $\Delta V = 10$ and $R_{max} = 120$.}
\label{Fig7Label}
\end{figure}

\begin{figure}[!htb]
\noindent \centering{}\includegraphics{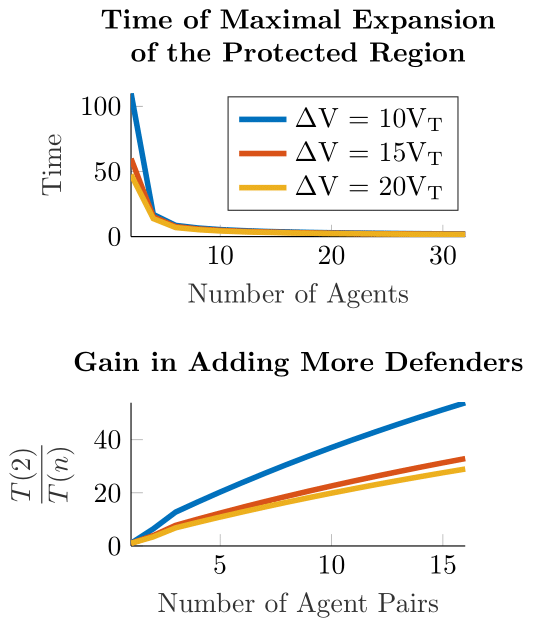} \caption{Time of maximal expansion of the protected region and gain in adding more defenders for different defender speeds. We simulated the circular defense pincer sweep processes for an even number of agents, ranging from $2$ to $32$ agents. The chosen values of the parameters are $r=10$, $V_T = 1$ and $R_{max} = 120$.}
\label{Fig8Label}
\end{figure}

\section{Spiral Defense Pincer Sweep Process}
\label{sec:spiral_defense}
In order to handle the inherent inefficiency of the circular defense pincer sweep protocol, which first and foremost results from the fact that at the beginning of each sweep only half the length of the defenders' sensors are inside the invader region, we propose a modification to the defense process that tackles this inefficiency. This modification strives to increase the part of the defenders' sensors over the invader region so that they can detect invaders further away from the protected region.
Therefore, a spiral scan in which the tip of a defender's sensor follows the expanding protected region's boundaries is proposed.  

\begin{figure}[!htb]
\noindent \centering{}\includegraphics[width=3.5in,height =2.464in]{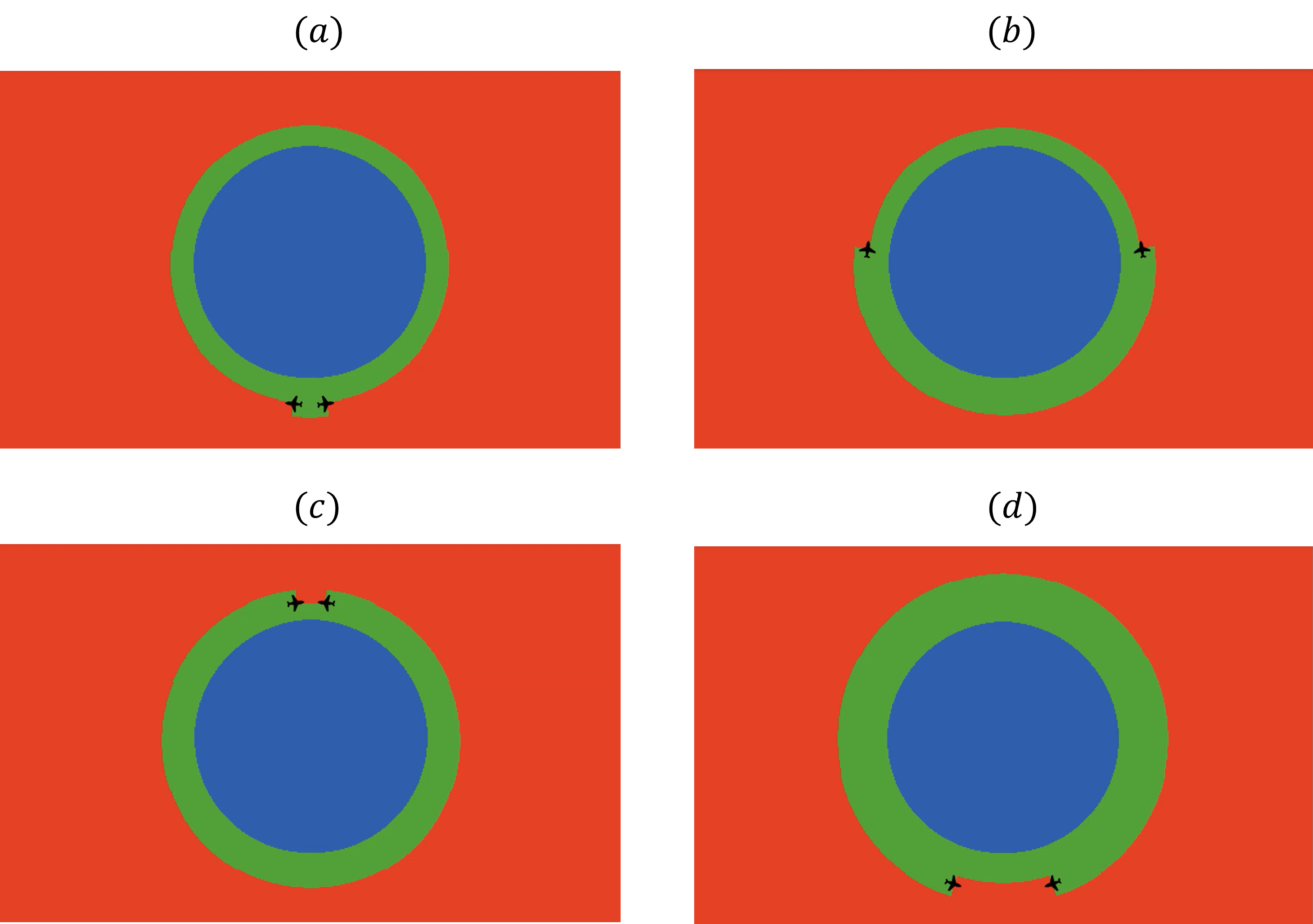} \caption{Swept areas and protected region's status for different times in a scenario where $2$ defenders perform the spiral defense pincer sweep process. (a) - Beginning of second sweep. (b) - Midway of the second sweep. (c) - End of the second sweep. (d) - End of the third sweep. Green areas show locations that were searched and hence do not contain invaders and red areas indicate locations where potential invaders may be present. Blue areas represent locations that belong to the initial protected region that does not contain invaders.}
\label{Fig9Label}
\end{figure}

\subsection{The Defense Task and Critical Speed Analysis}

At the beginning of the defense protocol, each defender's sensor overlaps the protected region by $0$ (and consequently a length of is $2r$ over the complementing invader region). Sweeping in a pincer movement enables defenders to have a critical speed that is based only upon the time it takes them to traverse their allocated angular section of $\frac{2\pi}{n}$. In a similar manner as in the circular defense pincer sweep process, defenders sweeping at greater speeds than the corresponding critical speed of the scenario, switch their search direction once they finish their outward advancement phase. At the next iteration defenders sweep around a section having a larger radius.

Defenders' begin their spiral motion with their sensors' tips tangent to the boundary of the protected region. To keep their sensors tangent to the protected region during the spiral sweeping phases, defenders move at an angle $\phi$ with respect to the normal of the protected region. $\phi$ depends on the ratio between defender and invader speeds. Moving by a constant angle $\phi$ with respect to the normal of the protected region allows to preserve the protected region's circular shape and to keep the entire length of the defenders' sensors outside of the protected region, thus enabling detection of invaders at greater distances from the boundary of the protected region. Fig. $1$ (b) depicts the starting locations of $2$ defenders performing the spiral defense pincer sweep process. Fig. $9$ shows the cleaning progress during the expansion of the protected region when $2$ defenders employ the spiral defense pincer sweep protocol. $\phi$ is given by,
\begin{equation}
\sin \phi  = \frac{{{V_T}}}{{{V_s}}}
\label{e1000}
\end{equation}
Hence, $\phi  = \arcsin \left( {\frac{{{V_T}}}{{{V_s}}}} \right)$. The defender's angular speed or rate of change of its angle with respect to the center of the protected region, $\theta _s$, can be expressed by the following  function of $\phi$,
\begin{equation}
\frac{{d{\theta _s}}}{{dt}} = \frac{{{V_s}\cos \phi }}{{{R_s}(t)}} = \frac{{\sqrt {{V_s}^2 - {V_T}^2} }}{{{R_s}(t)}}
\label{e1002}
\end{equation}
Integration of (\ref{e1002}) between the initial and final search times of a particular sweep leads to,
\begin{equation}
\int_0^{{t_\theta }} {\dot \theta \left( \zeta  \right)} d\zeta  = \int_0^{{t_\theta }} {\frac{{\sqrt {{V_s}^2 - {V_T}^2} }}{{  {R_0} + r - {V_T}\zeta}}d} \zeta 
\label{e36}
\end{equation}
Solving for $\theta \left( {{t_\theta }} \right)$ from  (\ref{e36}) results in,
\begin{equation}
\theta \left( {{t_\theta }} \right) = -\frac{{\sqrt {{V_s}^2 - {V_T}^2} }}{{{V_T}}}\ln \left( {\frac{{{R_0} + r - {V_T}{t_\theta }}}{{{R_0} + r }}} \right)  
\label{e15}
\end{equation}
Applying the exponent function to (\ref{e15}) yields,
\begin{equation}
\left( {{R_0} + r} \right){e^{-\frac{{{V_T}{\theta(t_\theta) }}}{{\sqrt {{V_s}^2 - {V_T}^2} }}}}  = R_0 + r - {V_T}{t_\theta }= {R_s}({t_\theta })
\label{e38}
\end{equation}

The time required for a defender to complete a spiral scan of the angular section under its responsibility corresponds to changing its angle $\theta$ by $\frac{2\pi}{n}$. The expansion of the invading region at this time must be less than or equal to $2r$, so that defenders will still be able to prevent the entrance of invaders to the protected region. Since during the defenders' outwards movements invaders may continue to progress toward the protected region, defenders' can only move outwards by a somewhat smaller distance to address this concern. If we were to neglect invaders' motion during the outward movement phases, the necessary requirement to ensure invaders cannot enter the protected region without being detected by the defenders is,
\begin{equation}
{R_0} - r \le {R_s}({t_\frac{2\pi}{n}})
\label{e39}
\end{equation}
Define,
\begin{equation}
\lambda = {e^{-\frac{{2\pi {V_T}}}{{n\sqrt {{V_s}^2 - {V_T}^2} }}}}
\label{e40}    
\end{equation}
Replacing ${R_s}({t_\frac{2\pi}{n}})$ with the expression for the defender's trajectory results in,
\begin{equation}
{R_0} - r \le \left( {{R_0} + r} \right)\lambda
\label{e42}
\end{equation}
Hence, to guarantee invaders cannot enter the protected region throughout the spiral scans without being detected, the defenders' speeds has to satisfy,
\begin{equation}
{V_S} \ge {V_T}\sqrt {\frac{{{{\left( {\frac{{2\pi }}{n}} \right)}^2}}}{{{{\left( {\ln \left( {\frac{{{R_0} + r}}{{{R_0} - r}}} \right)} \right)}^2}}} + 1}
\label{e43}
\end{equation}
In order to consider the progression of invaders toward the protected region and modify the critical speed in (\ref{e43}) to cope with this motion, defenders move outward from the protected region after completing a spiral sweep until meeting the invader region's wavefront travelling towards them with a speed of $V_T$ at the previous radius $R_0$. This consideration guarantees that all invaders are detected and that the critical speed of the spiral defense pincer sweep protocol is nearly equal to the optimal lower bound on the defender speed of Theorem $1$. Denote the expansion of the invader region in the first sweep by ${T_c}$. In order to ensure no invaders enter the region without being detected the following inequality must hold, ${V_T}{T_c} \leq \frac{{2r{V_s}}}{{{V_s} + {V_T}}}$. Replacing the expression for ${T_c}$ yields,
\begin{equation}
\left( {{R_0} + r} \right)\left(1- \lambda \right) \leq \frac{{2r{V_s}}}{{{V_s} + {V_T}}}
\label{e49}
\end{equation}

\begin{thm}
In the spiral defense pincer sweep process, the critical speed ,${V_s}$, enabling to successfully complete the defense task is obtained as the solution of,
\begin{equation}
{V_T}{T_c} = \frac{{2r{V_s}}}{{V_s} + {V_T}}
\label{e98996}
\end{equation}
Where ${T_c}$ equals,
\begin{equation}
{T_c} = \frac{{\left( {{R_0} + r} \right)\left( 1 -\lambda \right)}}{{{V_T}}}
\label{e98997}
\end{equation}
\end{thm}
The critical speed required to successfully perform the spiral defense pincer sweep strategy is computed by solving numerically the equation in Theorem $5$ with the Newton–Raphson method while using the critical speed in (\ref{e43}) as an initial guess. This speed, that guarantees all invaders are detected, is used in all further calculations of this section.

As shown in Fig. $10$. the spiral critical speed nearly equals the optimal lower bound, specifically for a small number of defenders. For example, when $2$ defenders perform the defense protocol the ratio between the spiral critical speed and the optimal lower bound on the speed is $1.06$. 

\begin{figure}[!htb]
\noindent \centering{}\includegraphics{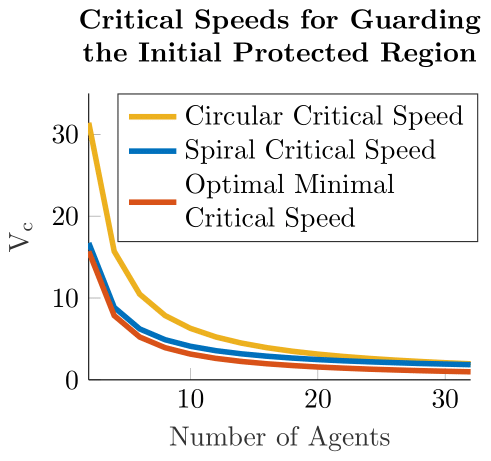} \caption{Critical speeds as a function of the number of defenders. The number of defenders is even, and ranges from $2$ to $32$ defenders, that perform the spiral defense pincer sweep protocol. The optimal lower bound on the critical speeds and the resulting critical speeds of the circular defense pincer sweep protocol are presented for comparison as well. The chosen values of the parameters are $r=10$, $V_T = 1$ and $R_0 = 100$.}
\label{Fig10Label}
\end{figure}

\subsection{The Maximal Expansion Task}
\subsubsection{Number of Sweeps Analysis}

\begin{thm}
The maximal radius that $n$ spiral sweeping defenders, with a linear sensor of length $2r$, a given speed $V_s$ and a maximal invader speed of $V_T$ can safely protect against the entrance of invaders is,
\begin{equation}
{\bar R_{{N_s}}} = \frac{{2r{V_s}}}{{\left( {1 - \lambda } \right)\left( {{V_s} + {V_T}} \right)}} - r
\label{e1004}
\end{equation}
\end{thm}

\begin{thm}
For a defender team with $n$ defenders for which $n$ is even, performing the spiral defense pincer sweep process, the number of sweeps required for the defender team to expand the protected region to a circle with a radius that is $\varepsilon$ close to the maximal boundable radius ${{\bar R}_{N_s}}$ is given by,
\begin{equation}
{N_n} = \left\lceil {\frac{{\ln \left( {\frac{{2r{V_s}\left( {1 - {V_T}} \right) - \varepsilon \left( {1 - \lambda } \right)\left( {{V_s} + 1} \right)\left( {{V_s} + {V_T}} \right)}}{{\left( {{V_s} + {V_T}} \right)\left( {\left( {{R_0} + r} \right)\left( {1 - \lambda } \right)\left( {{V_s} + 1} \right) - 2r{V_s}} \right)}}} \right)}}{{\ln \left( {\frac{{{V_T} + {V_s}\lambda  - 1 + \lambda }}{{{V_s} + {V_T}}}} \right)}}} \right\rceil
\label{e98992}
\end{equation}
Denote by ${T_{out}}(n)$ the sum of all outward advancement times and by ${T_{spiral}}(n)$ the sum of all the spiral search times. Hence, the total search time necessary for the defender team to expand the protected region to its maximal size is given by,
\begin{equation}
T(n) = {T_{out}}(n) + {T_{spiral}}(n)
\label{e98993}
\end{equation}
${T_{out}}(n)$ is given by,
\begin{equation}
\begin{array}{l}
{T_{out}}(n) = \frac{{2r\left( {{V_s} + 1 + {V_T}} \right)}}{{\left( {1 - \lambda } \right)\left( {{V_s} + {V_T}} \right)\left( {{V_s} + 1} \right)}} - \frac{{{R_0} + 2r + \varepsilon }}{{{V_s}}} - \frac{{r\left( {{V_s} - 1 + \lambda {V_s} + \lambda } \right)}}{{{V_s}\left( {1 - \lambda } \right)\left( {{V_s} + 1} \right)}}
\end{array}
\label{e769}
\end{equation}
And $T_{spiral}(n)$ is given by,
\begin{equation}
\begin{array}{l}
{T_{spiral}}(n) = \frac{{{R_0}{V_s} + {R_0}{V_T} + r{V_T} + 2r{V_s}{N_n} - r{V_s}}}{{{V_T}\left( {{V_s} + 1} \right)}}  \\ - \frac{{2r{V_s}\left( {{V_T} + {V_s}\lambda  - 1 + \lambda } \right)}}{{{V_T}{{\left( {{V_s} + 1} \right)}^2}\left( {1 - \lambda } \right)}}\\ -  \frac{{{V_s} + {V_T}}}{{\left( {1 - \lambda } \right)\left( {{V_s} + 1} \right)}}{\left( {\frac{{{V_T} + {V_s}\lambda  - 1 + \lambda }}{{{V_s} + {V_T}}}} \right)^{{N_n}}}\left( {\frac{{\left( {{R_0} + r} \right)\left( {1 - \lambda } \right)}}{{{V_T}}} - \frac{{2r{V_s}}}{{{V_T}\left( {{V_s} + 1} \right)}}} \right)
\end{array}
\label{e776}
\end{equation}

\end{thm}
\begin{proof}
Denote by $\Delta V >0$ the excess speed of the defender above the critical speed. Hence, the defender's speed is, $V_s = V_c + \Delta V$. At the start of each sweep the center of the defender's sensor is at a distance of ${R_i} + r$ from the protected region's  center. $\theta \left( {{t_\theta }} \right)$ is calculated in (\ref{e15}). Substituting $R_0$ with $R_i$ results in,

\begin{equation}
\theta \left( {{t_\theta }} \right) = -\frac{{\sqrt {{V_s}^2 - {V_T}^2} }}{{{V_T}}}\ln \left( {\frac{{{R_i} + r - {V_T}{t_\theta }}}{{{R_i} + r }}} \right) 
\label{e46}
\end{equation}

Denote the time required for a defender to complete the search of an angular section of $\theta \left( {{t_\theta }} \right)=\frac{2\pi}{n}$ by ${T_{spiral}}_i$. It is obtained from (\ref{e46}) and equals,
\begin{equation}
{T_{spiral}}_{_i} = \frac{{\left( {{R_i} + r} \right)\left( {1 - \lambda} \right)}}{{{V_T}}}
\label{e47}
\end{equation}

If defenders move with speeds greater than the critical speed required for the scenario, the distance a defender may advance outward from the center of the protected region is ${\delta _i}(\Delta V)$,
\begin{equation}
{\delta _i}(\Delta V) = 2r - {V_T}{T_{spiral}}_i \hspace{1mm}, \hspace{1mm} 0  \le {\delta _i}(\Delta V) \le 2r
\label{e1015}
\end{equation}
Once defenders finish the outward advancement phase, the protected region expands to an updated circular protected region with a radius of ${R_{i + 1}} = {R_i} + {\delta _i}(\Delta V)$. At the end of the spiral maneuver the protected region is again circular, with a larger radius. The proof for this property follows similar steps as provided in Appendix $H$ of \cite{francos2021swarms}.

Depending on the number of participating defenders and the iteration number, the distance a defender can advance outwards after completing a sweep is,
\begin{equation}
{\delta _i}(\Delta V) = 2r - \left({{R}_i +r}\right)\left( {1 - \lambda} \right)
\label{e1067}
\end{equation}
Where in the term ${\delta _i}(\Delta V)$, $\Delta V$ denotes the increase in the agent's speed relative to the critical speed. The number of sweep iterations the defenders performed around the protected region is denoted by $i$, where $i$ starts from sweep number $0$.

Since during the time in which defenders move outwards from the protected region, invaders continue to advance towards it, defenders can advance outwards by a slightly lesser distance than ${\delta _i}(\Delta V)$. The time required for defenders to move outwards until their entire sensors are outside of the protected region depends on the relative speed between the defenders' outward speed and the invader region's inwards expansion speed. Therefore, defenders are able to advance outwards by,
\begin{equation}
{\delta _{{i_{eff}}}}(\Delta V) = {\delta _i}(\Delta V)\left( {\frac{{{V_s}}}{{{V_s} + {V_T}}}} \right)
\label{e700}
\end{equation}
Hence, the radius of the expanded circular protected region is,
\begin{equation}
{R_{i + 1}} = {R_i} + {\delta _i}(\Delta V)\left( {\frac{{{V_s}}}{{{V_s} + {V_T}}}} \right)
\label{e1097}
\end{equation}
Denote $\widetilde {R}_i = R_i + r$. Replacing $R_i$ with $\widetilde {R}_i$ results in a similar structure of formulas as in the circular defense pincer sweep process and enables to use the same methodology along with the appropriate change of coefficients to solve for the maximal defendable radius of the protected region.
Replacing the expression for ${{\delta _i}(\Delta V)}$ into (\ref{e1097}) yields,
\begin{equation}
{{\tilde R}_{i + 1}} = {{\tilde R}_i} + \left( {2r - {{\tilde R}_i}\left( {1 - \lambda} \right)} \right)\frac{{{V_s}}}{{{V_s} + {V_T}}}
\label{e1072}
\end{equation}
Rearranging terms results in a difference equation that resembles the equation obtained for the circular defense pincer sweep process,
\begin{equation}
{{\tilde R}_{i + 1}} = {{\tilde R}_i}\left( {\frac{{{V_T} + {V_s}\lambda}}{{{V_s} + {V_T}}}} \right) + \frac{{2r{V_s}}}{{{V_s} + {V_T}}}
\label{e1073}
\end{equation}
Denote the coefficients in (\ref{e1073}) by,
\begin{equation}
{c_1} =  \frac{{2r{V_s}}}{{{V_s} + {V_T}}}, {c_2} = \frac{{{V_T} + {V_s}\lambda}}{{{V_s} + {V_T}}}
\label{e1075}
\end{equation}
Hence, (\ref{e1073}) is expressed as,
\begin{equation}
\widetilde{R}_{i + 1} = {c_2}\widetilde{R}_i +{c_1}
\label{e345}    
\end{equation}

Since the defenders need to move outwards from protecting a region with a smaller radius, and during this outwards movement the protected region continues to shrink due to possible movements of invaders, the defenders can protect a slightly smaller region. For any even number of defenders, $n$, the expansion protocol continues in this way until the protected region is enlarged to the largest possible circle. Let $\varepsilon>0$ and denote by ${\widehat{R}_{{N_n} - 1}}$ the radius of the protected region that is $\varepsilon$ close to ${{\bar R}_{N_s}}$,
\begin{equation}
\widehat{R}{_{{N_n} - 1}} = R_{max} = \frac{{2r}}{{1 - \lambda }} - r - \varepsilon 
\label{e901}
\end{equation}

Due to the same difference equation structure as in the circular defense pincer sweep protocol, the number of spiral sweeps is calculated similarly. Hence, the number of iterations required for the defender team to expand the protected region to a circle of radius $\widehat{R}{_{{N_n} - 1}}$ is,

\begin{equation}
{N_n} = \left\lceil {\frac{{\ln \left( {\frac{{ - \left( {r + \varepsilon } \right)\left( {1 - \lambda } \right)}}{{\left( {{R_0} + r} \right)\left( {1 - \lambda } \right) - 2r}}} \right)}}{{\ln \left( {\frac{{{V_T} + {V_s}\lambda }}{{{V_s} + {V_T}}}} \right)}}} \right\rceil 
\label{e1076}
\end{equation}

The radius $R_{max} =\widehat{R}{_{{N_n} - 1}}$ is the maximal radius the protected region expands to, and is used to calculate the number of sweeps required to reach this radius. The actual radius of the protected region after $N_n$ sweeps is denoted by ${R}{_{{N_n} - 1}}$ and is computed after $N_n$ is calculated. After the last spiral sweep, the defenders perform the last outward advancement, and the defenders continue to perform spiral sweeps around a protected region with a radius of $R_{max}$.

\subsubsection{Outward Advancement Times Calculation}

The outward advancement time depends on the iteration number. It is denoted by ${T_{ou{t_i}}}$ and is expressed as,
\begin{equation}
{T_{ou{t_i}}} = \frac{{{\delta _{{i_{eff}}}}(\Delta V)}}{{{V_s}}} = \frac{{2r - {{\tilde R}_i}\left( {1 - \lambda} \right)}}{{{V_s} + {V_T}}}
\label{e1077}
\end{equation}

Denote the total outward advancement times until the protected region is enlarged to a circle of radius ${R}{_{{N_n} - 1}}$ by $\widetilde{T}_{out}$, where $\widetilde{T}_{out}(n) = \sum\limits_{i = 0}^{{N_n}-2} {T_{ou{t_i}}}$. 

Throughout the outward advancement phases the defenders
do not perform sweeping and detection of invaders and hence
invaders are not detected until the defenders finish their outwards motion and resume the sweeping of the protected
region. Following a defender’s completion of the outwards
progression phase, its sensor overlaps the invader region
$2r$ (and therefore the footprint of its sensor that is over the protected region is $0$). The total search time until the protected region is expanded into a circle of radius ${R}{_{{N_n} - 1}}$ is given by the sum of the total spiral and outward advancement sweep times. Hence,
\begin{equation}
T(n) = \widetilde{T}_{out}(n) + T_{spiral}(n)
\label{e1083}
\end{equation}
Replacing the expression for ${T_{ou{t_i}}}$ yields that the accumulative outward advancement times before the protected region is enlarged to its maximal size are,
\begin{equation}
\widetilde{T}_{out}(n) = \sum\limits_{i = 0}^{{N_n} - 2} {{T_{ou{t_i}}} = } \frac{{2r\left( {{N_n} - 1} \right)}}{{{V_s} + {V_T}}} - \frac{{\left( {1 - \lambda} \right)\sum\limits_{i = 0}^{{N_n} - 2} {{{\tilde R}_i}} }}{{{V_s} + {V_T}}}
\label{e702}
\end{equation}
The full derivation of $\widetilde{T}_{out}(n) = \sum\limits_{i = 0}^{{N_n} - 2} {{T_{ou{t_i}}}}$ is continued in Appendix $B$. Hence,

\begin{equation}
\begin{array}{l}
\widetilde{T}_{out}(n) = \frac{{2r}}{{{V_s} + {V_T}}} - \frac{{{{\tilde R}_0}}}{{{V_s}}} + \frac{{2r\left( {{V_T} + {V_s}\lambda } \right)}}{{{V_s}\left( {{V_s} + {V_T}} \right)\left( {1 - \lambda } \right)}} +\\ {\left( {\frac{{{V_T} + {V_s}\lambda }}{{{V_s} + {V_T}}}} \right)^{{N_n} - 1}}\left( {\frac{{{{\tilde R}_0}}}{{{V_s}}} - \frac{{2r}}{{{V_s}\left( {1 - \lambda } \right)}}} \right)
\end{array}
\label{e1084}
\end{equation}

${R}{_{{N_n} - 1}}$ is computed by the same methodology as in the circular defense sweep process section and is given by,
\begin{equation}
\begin{array}{l}
{R_{{N_n} - 1}} = \frac{{r\left( {1 + \lambda } \right)}}{{1 - \lambda }} + {\left( {\frac{{{V_T} + {V_s}\lambda }}{{{V_s} + {V_T}}}} \right)^{{N_n} - 1}}\left( {{R_0} + r - \frac{{2r}}{{1 - \lambda }}} \right)
 \end{array}
\label{e782}
\end{equation}
Following the last spiral sweep, the defenders perform the last outward advancement, until reaching $R_{max}$. The time it takes them to perform this last outward advancement is,
\begin{equation}
T_{out_{last}} = \frac{\left(\widehat{R}{_{{N_n} - 1}} - {R_{{N_n} - 1}}\right)}{V_s}
\label{e1103}
\end{equation}
After this last outward sweep, the defenders perform spiral sweeps around a protected region of radius $R_{max}$ and continuously protect the region from the entrance of invaders after reaching the maximal protected region they can guard.
Summing $\widetilde{T}_{out}(n)$ and the last outward advancement time in (\ref{e1103}) yields,
\begin{equation}
{T_{out}}(n) = \frac{{\lambda \left( {{R_0} + r + \varepsilon } \right) + r - {R_0} - \varepsilon }}{{{V_s}\left( {1 - \lambda } \right)}}
\label{e1104}
\end{equation}
\subsubsection{Spiral Sweep Times Calculation}
The time to perform a spiral sweep around radius $\widetilde{R}_i$ is calculated by multiplying $\widetilde{R}_i$ with $\frac{1 -\lambda }{V_T}$. Therefore, by multiplying (\ref{e1073}) with  $\frac{ 1 - \lambda}{V_T}$ the following difference equation for the sweep spiral times is,

\begin{equation}
{T_{i + 1}} = {c_2}{T_i} + {c_3}
\label{e37}
\end{equation}
Where the coefficient $c_3$ is,
\begin{equation}
{c_3} = \frac{{2r{V_s}\left( {1 - \lambda } \right)}}{{\left( {{V_s} + {V_T}} \right){V_T}}}
\label{e1062}
\end{equation}
The total spiral sweep times required to expand the protected region into its largest size are calculated by similar steps as the circular sweep times in the previous section. Hence,
\begin{equation}
T_{spiral}(n) = \frac{{{T_0} - {c_2}{T_{N_n - 1}} + \left( {{N_n} - 1} \right){c_3}}}{{1 - {c_2}}}
\label{e100}
\end{equation}
The time required for defenders to perform the first spiral sweep is,
\begin{equation}
{T_0} = \frac{{\left( {{R_0} + r} \right)\left( {1 - \lambda } \right)}}{{{V_T}}}
\label{e44}
\end{equation}
The time to perform the last spiral sweep before the protected region reaches its maximal radius of ${R}{_{{N_n} - 1}}$ is given by,
\begin{equation}
{T_{N_n - 1}} = \frac{{{c_3}}}{{1 - {c_2}}} +  {{c_2}^{{N_n} - 1}} \left( {{T_0} - \frac{{{c_3}}}{{1 - {c_2}}}} \right)
\label{e1089}
\end{equation}
Substitution of coefficients results in,
\begin{equation}
\begin{array}{l}
{T_{{N_n} - 1}} = \frac{{2r}}{{{V_T}}} + {\left( {\frac{{{V_T} + {V_s}\lambda }}{{{V_s} + {V_T}}}} \right)^{{N_n} - 1}}\left( {\frac{{{R_0}\left( {1 - \lambda } \right) - r\left( {1 + \lambda } \right)}}{{{V_T}}}} \right)
 \end{array}
\label{e729}
\end{equation}
Yielding that the total spiral sweep times are,
\begin{equation}
\begin{array}{l}
{T_{spiral}}(n) = \frac{{\left( {{R_0} + r} \right)\left( {{V_s} + {V_T}} \right)}}{{{V_T}{V_s}}} - \frac{{2r\left( {{V_T} + {V_s}\lambda } \right)}}{{{V_T}{V_s}\left( {1 - \lambda } \right)}} - \\ \frac{{{V_s} + {V_T}}}{{{V_s}\left( {1 - \lambda } \right)}}{\left( {\frac{{{V_T} + {V_s}\lambda }}{{{V_s} + {V_T}}}} \right)^{{N_n}}}\left( {\frac{{{R_0}\left( {1 - \lambda } \right) - r\left( {1 + \lambda } \right)}}{{{V_T}}}} \right) + \frac{{2r\left( {{N_n} - 1} \right)}}{{{V_T}}}
\end{array}
\label{e1086}
\end{equation}

\end{proof}

\subsubsection{Numerical Experiments}

Fig. $11$ presents the maximal protected region's radius that the defenders are able to protect. The maximal radius clearly depends on the number of defenders and their speeds. Fig. $12$ presents the number of sweeps required to expand the protected region to its maximal size as a function of $\varepsilon$. Fig. $13$ presents 
the expansion time of the protected region to $R_{N_s} - \varepsilon$ for a fixed speed exceeding the spiral critical speed of $2$ defenders that perform the spiral defense pincer sweep protocol. Fig. $14$ presents the total search times for different numbers of defenders. In all presented graphs the defenders' speed is equal and is independent of the number of defenders performing the expansion protocol, and is chosen so that search times of defender teams with different number of defenders are correctly compared. The values of $\Delta V$ mentioned in the plots are speeds above the critical speed of two defenders employing the spiral defense pincer sweep process. The second plot from the top of Fig. $14$. presents the search time reduction obtained when the number of participating defenders increases. The critical speed required for the defender team to perform the defense task is determined by solving numerically the equation presented in Theorem $5$, consequently ensuring invaders cannot enter the protected region undetected. 

\begin{figure}[!htb]
\noindent \centering{}\includegraphics{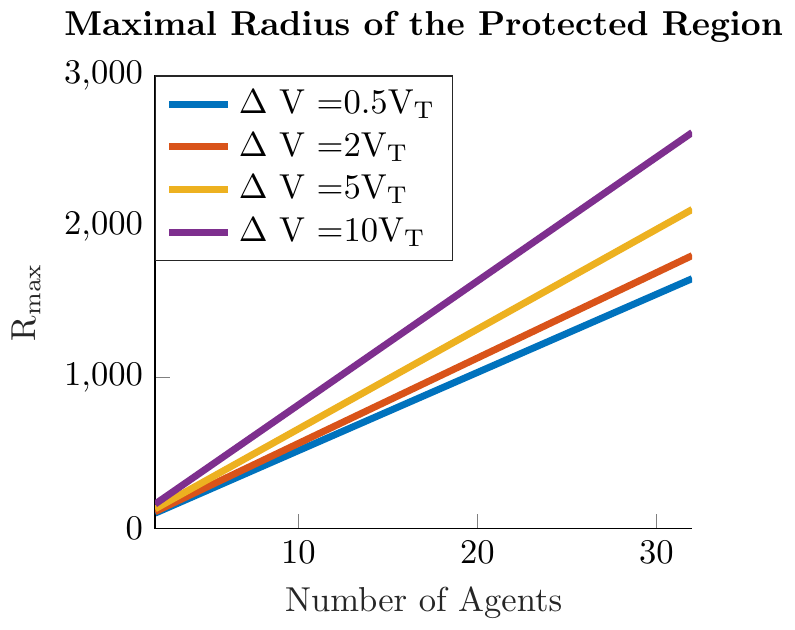} \caption{Maximal protected region's radius. The spiral defense pincer sweep processes were simulated with an even number of defenders, ranging from $2$ to $32$ defenders. The chosen values of the parameters are $r=10$, $V_T = 1$ and $\varepsilon = 0.2$.}
\label{Fig11Label}
\end{figure}

\begin{figure}[!htb]
\noindent \centering{}\includegraphics{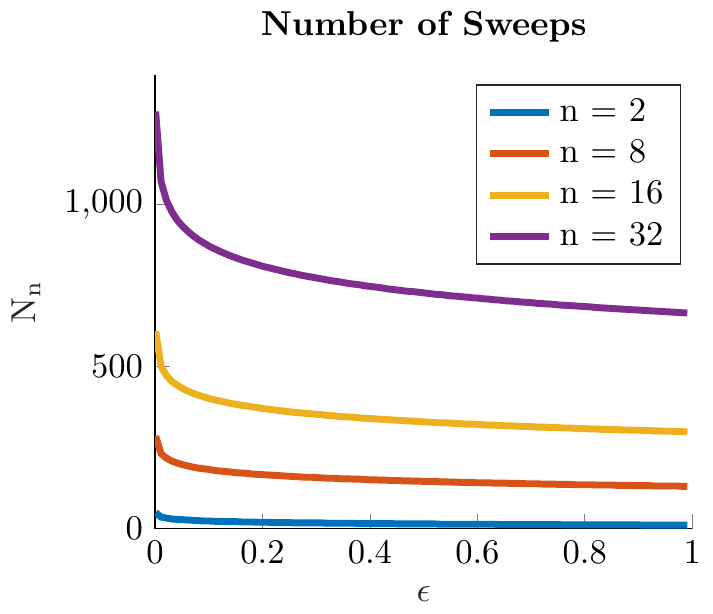} \caption{Number of sweeps required to expand the protected region to its maximal size as a function of $\varepsilon$. We plot the results for defenders performing the spiral defense pincer sweep processes with $2,8,16$ and $32$ defenders. The chosen values of the parameters are $r=10$, $V_T = 1$, $V_s = 17.2219$.}
\label{Fig12Label}
\end{figure}

\begin{figure}[!htb]
\noindent \centering{}\includegraphics{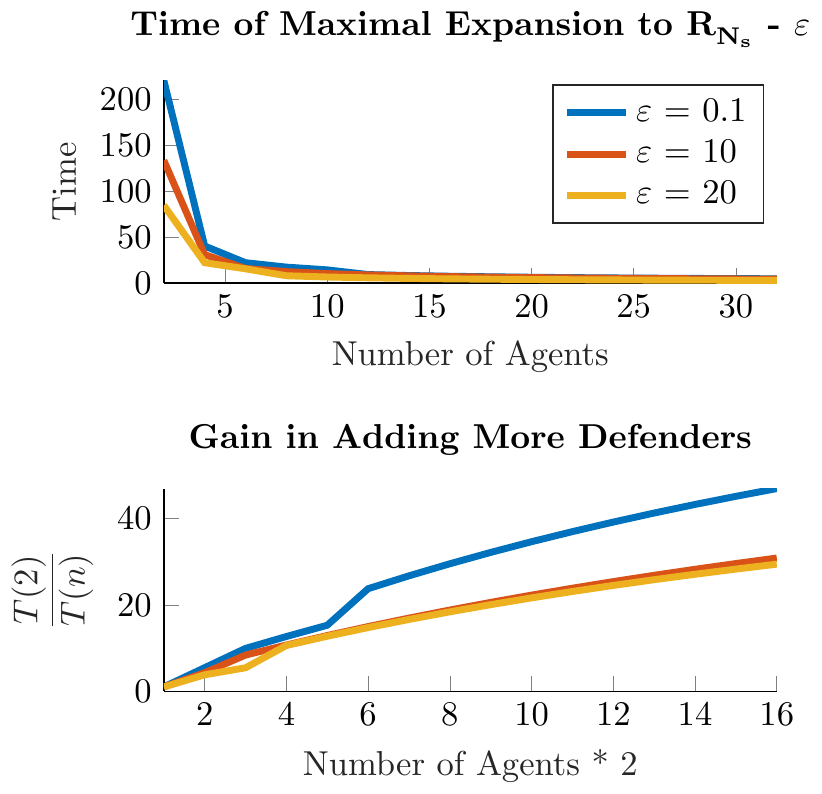} \caption{Time of maximal expansion of the protected region to $R_{N_s} - \varepsilon$ and gain in adding more defenders for equal defender speeds. We simulated the spiral defense pincer sweep processes for an even number of defenders, ranging from $2$ to $32$ defenders. The chosen values of the parameters are $r=10$, $V_T = 1$, $\Delta V = 10$ and $R_{max} = 150$.}
\label{Fig13Label}
\end{figure}

\begin{figure}[!htb]
\noindent \centering{}\includegraphics{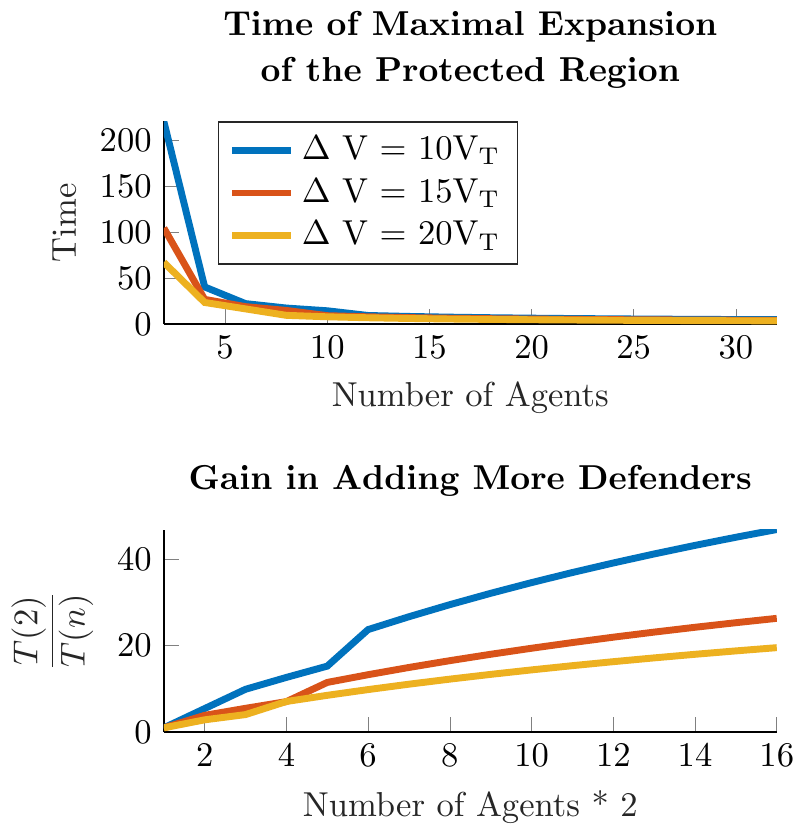} \caption{Time of maximal expansion of the protected region and gain in adding more defenders for different defender speeds. We simulated the spiral defense pincer sweep processes for an even number of defenders, ranging from $2$ to $32$ defenders. The chosen values of the parameters are $r=10$, $V_T = 1$ and $R_{max} = 150$.}
\label{Fig14Label}
\end{figure}

\section{Comparative Analysis Between Circular and Spiral Defense Pincer Sweep Strategies}
\label{sec:comparative_analysis}
The purpose of this section is to compare between the attained results for the circular and spiral defense pincer sweep processes using the relevant performance metrics. These metrics constitute the minimal defender speed required for successful defense of the initial protected region, the time to expand the protected region to the maximal defendable area and the maximal feasible protected region's radius resulting from the defense protocol. To accurately compare between the total search times of defender swarms that can perform both the circular and spiral defense pincer sweep processes, the number of defenders as well as the defenders' speed has to be equal in the compared circular and spiral defender swarms.

Defenders performing the circular defense pincer sweep process require a higher critical speed compared to  defenders performing the the spiral defense pincer sweep process. Therefore, Fig. $15$ presents the maximal protected region's radius that the defender team is able to expand the region into, when defenders employ the spiral defense pincer sweep process. The results are obtained for different speeds above the circular critical speed. The resulting maximal radius is clearly larger compared to the maximal protected region's radius that is achieved with a defender team that employs the circular defense pincer sweep process in Fig. $5$.

Fig. $16$ shows the spiral defense pincer sweep process's total search times obtained for different speeds above the circular critical speed of $2$ defenders. This implies that values of $ \Delta V$ shown in the plots correspond to defender speeds that equal nearly twice the spiral critical speeds. Requiring a higher critical speed means defender teams performing the circular defense pincer sweep process can expand the protected region to a smaller area compared to a defender team with the same capabilities performing the spiral protocol.

Fig. $17$ compares the search times until the maximal expansion of the protected region is obtained and the gain in adding more defenders for circular sweeping swarms and spiral sweeping swarms. The results are computed with the same defender speeds for both the circular and spiral defense pincer sweep processes. The reduction in total search time achieved when defenders perform the spiral defense pincer sweep process are clearly noticeable. This result holds regardless to the number of defenders that perform the defense protocols or to their speeds.

\begin{figure}[!htb]
\noindent \centering{}\includegraphics{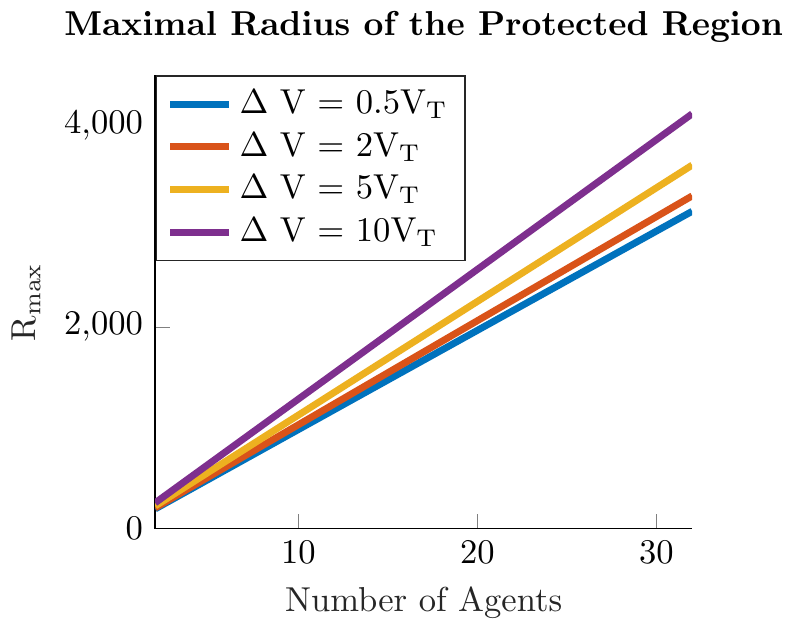} \caption{Maximal protected region's radius. We simulated the spiral sweep protocols with an even number of defenders, ranging from $2$ to $32$ defenders. We show results obtained for different values of speeds above the critical speed of $2$ defenders that employ the circular defense pincer sweep process. The chosen values of the parameters are $r=10$, $V_T = 1$ and $\varepsilon = 0.2$.}
\label{Fig15Label}
\end{figure}

\begin{figure}[!htb]
\noindent \centering{}\includegraphics{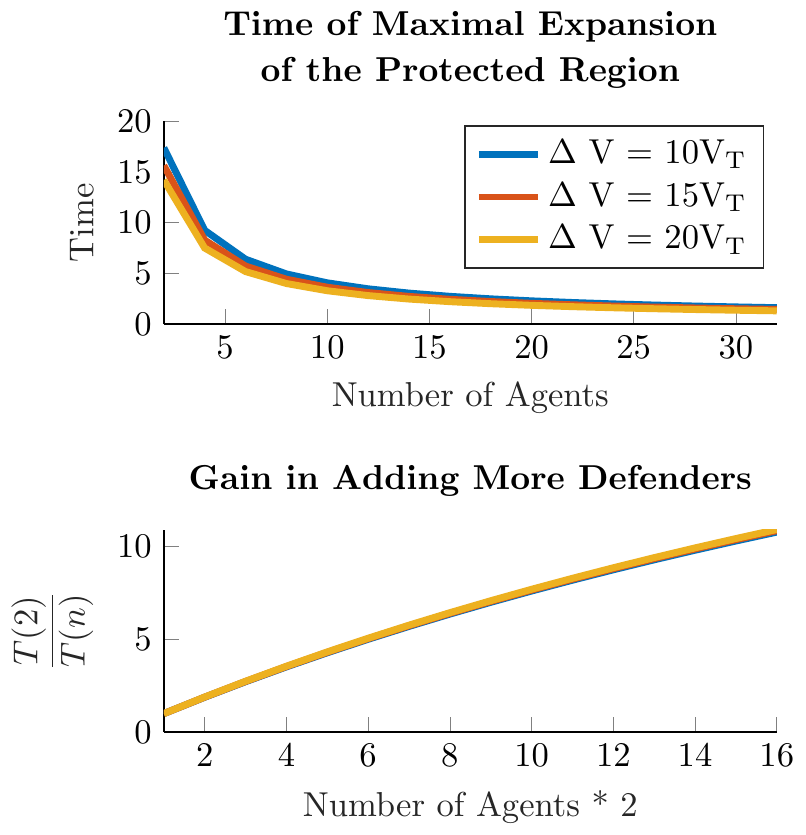} \caption{Time of maximal expansion of the protected region and gain in adding more defenders for different defender speeds above the critical speed of $2$ defenders employing the circular defense pincer sweep process. We simulated the spiral pincer sweep protocols for an even number of defenders, ranging from $2$ to $32$ defenders. The chosen values of the parameters are $r=10$, $V_T = 1$ and $R_{max} = 120$.}
\label{Fig16Label}
\end{figure}

\begin{figure}[!htb]
\noindent \centering{}\includegraphics[width=3.5in,height =3.5in]{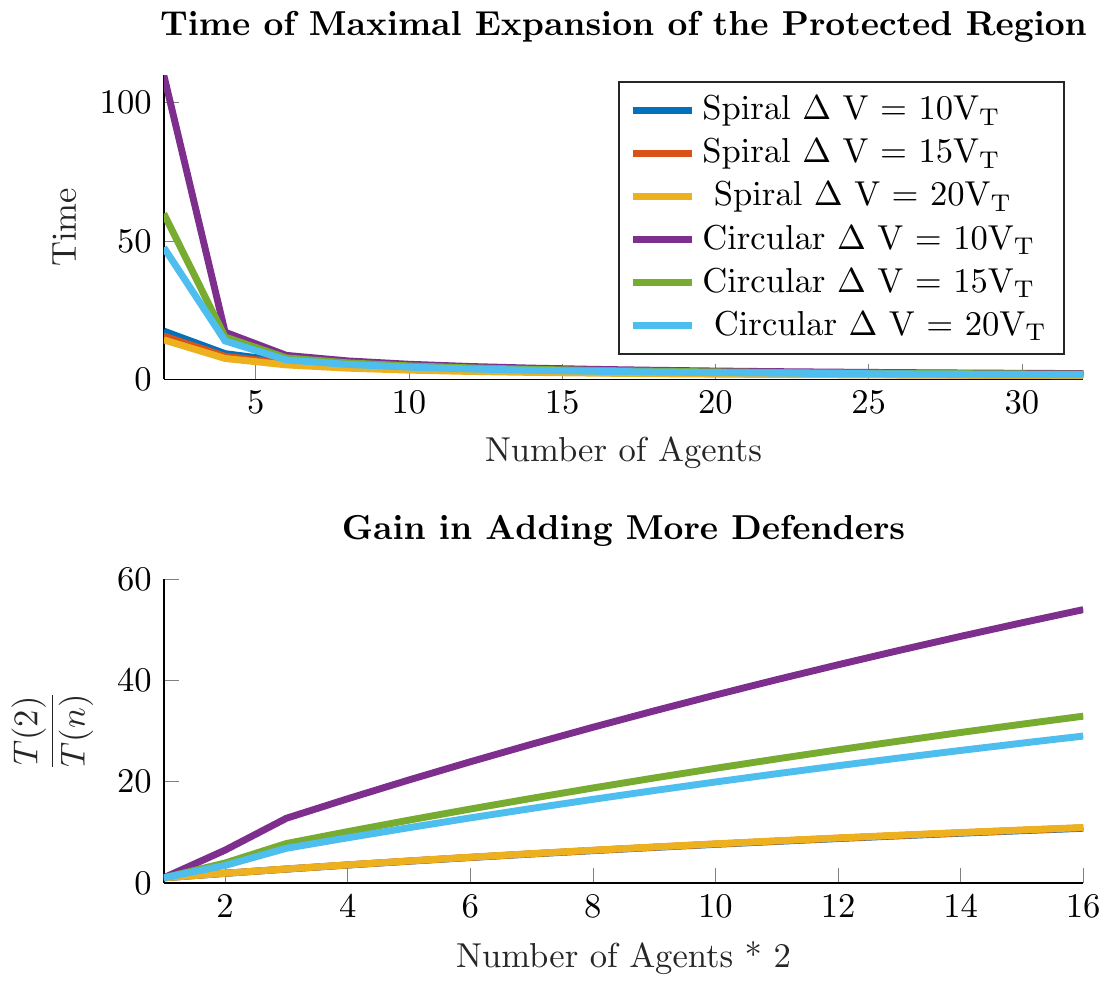} \caption{Time of maximal expansion of the protected region and gain in adding more defenders for the circular and spiral defense sweep protocols. We simulated sweep protocols with an even number of defenders, ranging from $2$ to $32$ defenders, that perform the circular and spiral defense pincer sweep protocols at speeds above the critical speed of $2$ defenders that perform the circular defense pincer sweep process. The chosen values of the parameters are $r=10$, $V_T = 1$ and $R_{max} = 120$.}
\label{Fig17Label}
\end{figure}

\section{Comparison to State-of-the-Art Same-Direction Defense Strategies}
\label{sec:comparison_to_same_direction_protocols}
The purpose of this section is to compare the developed circular and spiral pincer sweep guarding and expansion strategies to prevalent approaches for defense against smart invaders which are considered as the state-of-the-art in defense against smart invaders. Such approaches usually distribute the defending agents equally around the protected region and require that all defenders move in the same-direction. Such an approach is presented in \cite{mcgee2006guaranteed}, although the authors are interested only on solving the defense task and do not provide explicit expansion protocols that allow to achieve a maximal protected region or a detailed analytical analysis of sweep times. Hence, we develop two alternative same-direction defense protocols, circular and spiral, that enable the comparison of pincer-based and same-direction defense protocols against smart invaders.

We provide a quantitative comparison between the discussed $3$ metrics: critical speeds, maximal defendable area and the time required to reach maximal expansion. Circular and spiral defense pincer sweep protocols and circular and spiral defense same-direction sweep protocols are compared, proving the superiority of pincer-based approaches across all $3$ metrics. We prove that the corresponding pincer-based protocols yield lower critical speeds, shorter time to increase the protected region to its maximal size as well as the ability to expand the protected region to a larger area compared to same-direction protocols. 

These results are expected since defenders implementing same-direction protocols need to scan additional angular portions of the environment in each sweep around the region, to ensure no invader enters the protected region undetected. However, in pincer-based defense protocols, scanning such additional sectors is not required as a result of the complementary trajectories implemented by the defenders.

The critical speed necessary for defenders that perform the same-direction circular or spiral defense sweep protocols is higher compared to the minimal critical speed of their pincer-based counterparts. This can be observed in Fig. $18$. The same-direction defense protocols are developed by using similar considerations as the same-direction protocols in \cite{francos2021pincer}. These considerations lead to a same-direction circular protocol speed that equals,

\begin{equation}
{V_{{c_{circ\_same}}}} = \frac{{2\pi {R_0}{V_T}}}{{nr}} + {V_T}
\label{e150}
\end{equation}
The solution for the spiral same-direction defense protocol critical speed is solved numerically using the Newton–Raphson method from the equation below while using the spiral pincer-based critical speed as an initial guess.

\begin{equation}
\begin{array}{l}
F\left( {{V_s}} \right) = \frac{{2r{V_s}}}{{{V_s} + {V_T}}} \\ - \left( {{R_0} + r} \right)\left( {1 - {e^{ - \frac{{\left( {\frac{{2\pi }}{n} + \arcsin \left( {\frac{{2r{V_s}}}{{\left( {{V_s} + {V_T}} \right)\left( {{R_0} + 2r} \right)}}} \right)} \right){V_T}}}{{\sqrt {{V_s}^2 - {V_T}^2} }}}}} \right)
\end{array}
\label{e151}
\end{equation}
The angle $\beta_0$ denotes an additional angular sector that needs to be guarded in addition to the $\frac{2\pi}{n}$ angular sector that is swept when performing the spiral defense pincer protocols. $\beta_0$ is given by,

\begin{equation}
{\beta _0} = \arcsin \left( {\frac{{2r{V_s}}}{{\left( {{V_s} + {V_T}} \right)\left( {{R_0} + 2r} \right)}}} \right)
\label{e152}
\end{equation}

\begin{figure}[!htb]
\noindent \centering{}\includegraphics{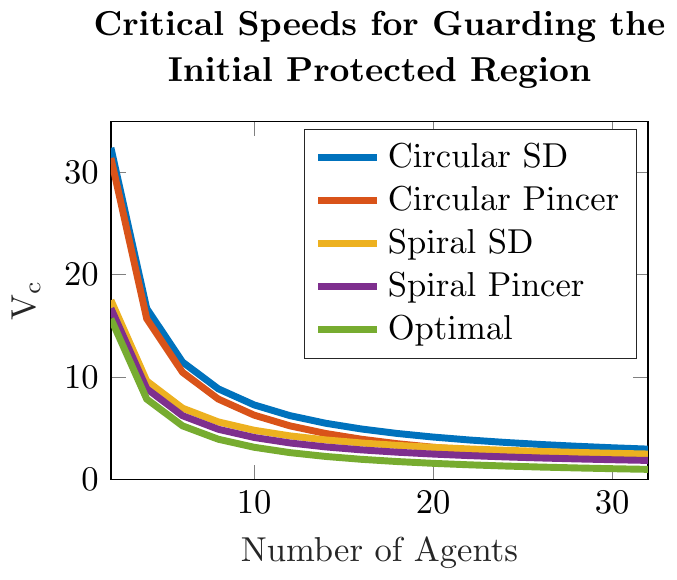} \caption{Critical speeds as a function of the number of defenders. The number of defenders is even, and ranges from $2$ to $32$ defenders, that perform the spiral and circular defense pincer sweep protocols as well as the spiral and circular defense same-direction sweep protocols. The optimal lower bound on the critical speeds is presented for comparison as well. The chosen values of the parameters are $r=10$, $V_T = 1$ and $R_0 = 100$.}
\label{Fig18Label}
\end{figure}

Because pincer-based defense sweep protocols require a lower critical speed compared to same-direction defense strategies, to fairly compare the performance of the different defense protocols, all defenders in each of the compared swarms is required to move at speeds above the critical speed of $2$ defenders that perform the same-direction circular defense protocol since it has the highest critical speed compared to the circular pincer, spiral pincer and spiral same-direction defense protocols.

The necessity to have a higher critical speed means that there are domains that can be successfully guarded using a team of defenders implementing pincer-based defense protocols but cannot be defended with a team of equal capabilities that performs same-direction defense protocols. Additionally, this means that defender teams performing pincer-based defense protocols are able to expand the protected region into a larger area compared to their same-direction alternative protocols.

Fig. $19$ shows the maximal defendable protected region's radius attained for each defense protocol. The results show that the spiral pincer-based approaches are best while circular same-direction defense protocols allow defenders to expand the protected region to the smallest area compared to the other expansion algorithms.

\begin{figure}[!htb]
\noindent \centering{}\includegraphics{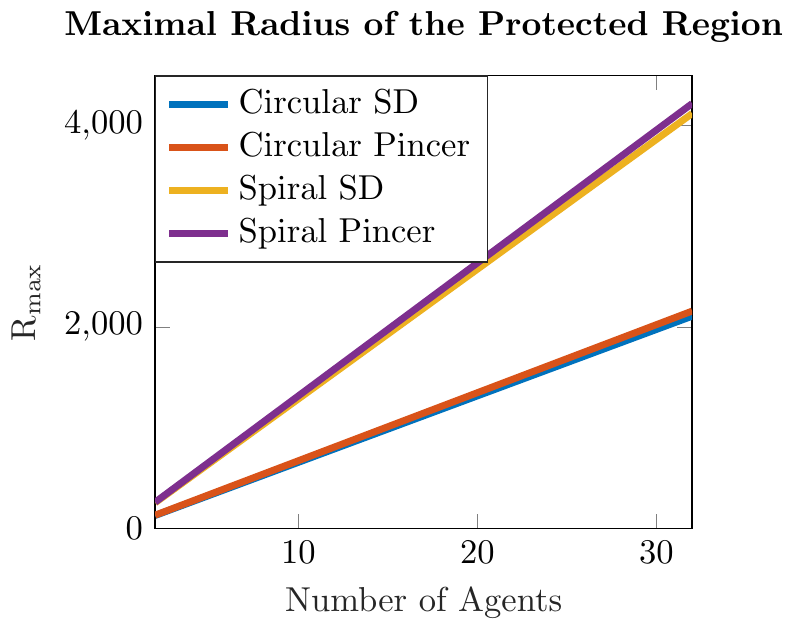} \caption{Maximal protected region's radius. We simulated the spiral and circular defense pincer sweep protocols as well as the spiral and circular defense same-direction sweep protocols with an even number of defenders, ranging from $2$ to $32$ defenders. We show results obtained for $\Delta V = 10 V_T$ above the critical speed of $2$ defenders that employ the circular defense same-direction sweep process. The chosen values of the parameters are $r=10$, $V_T = 10$ and $\varepsilon = 0.2$.}
\label{Fig19Label}
\end{figure}

Fig. $20$ shows the time until maximal expansion of the protected region to a radius of $R_{max} = 120$ for circular and spiral same-direction and circular and spiral pincer-based protocols. The value of $R_{max} = 120$ was chosen since a region with this radius can be successfully guarded by all $4$ protocols. All compared swarms have equal number of defenders and move at speeds that are $10 V_T$ above the critical speed of $2$ defenders that perform the circular defense same-direction sweep process. Results show that the spiral defense pincer sweep protocol results in the fastest expansion time of the protected region to a given radius.

\begin{figure}[!htb]
\noindent \centering{}\includegraphics{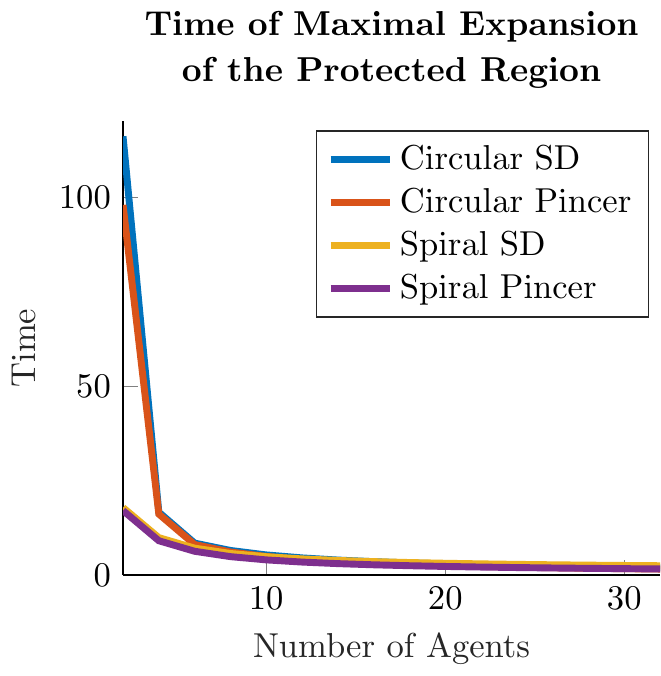} \caption{Time of maximal expansion of the protected region for the circular and spiral defense sweep protocols as well as the spiral and circular defense same-direction sweep protocols. We simulated sweep protocols with an even number of defenders, ranging from $2$ to $32$ defenders, that perform the defense sweep protocols at speeds of $10 V_T$ above the critical speed of $2$ defenders that perform the circular defense same-direction sweep process. The chosen values of the parameters are $r=10$, and $R_{max} = 120$.}
\label{Fig20Label}
\end{figure}

Fig. $21$ shows a zoomed-in plot of Fig. $20$ that displays the time until maximal expansion of the protected region, for swarms of defenders with $4$ to $22$ defenders. Results show that the spiral pincer defense protocols enables the defending team to expand the protected region to an area of a certain size more quickly compared to the circular same-direction, circular pincer and spiral same-direction defense protocols. Additionally, results show that for increasing number of defenders, circular pincer-based protocols lead to shorter sweep times even when compared to spiral same-direction defense protocols. This implies that despite the fact that pincer-based circular defense protocols may be implemented with defenders possessing more basic capabilities compared to defenders executing spiral strategies, the cooperation among defenders greatly improves the performance of the defender team.

\begin{figure}[!htb]
\noindent \centering{}\includegraphics{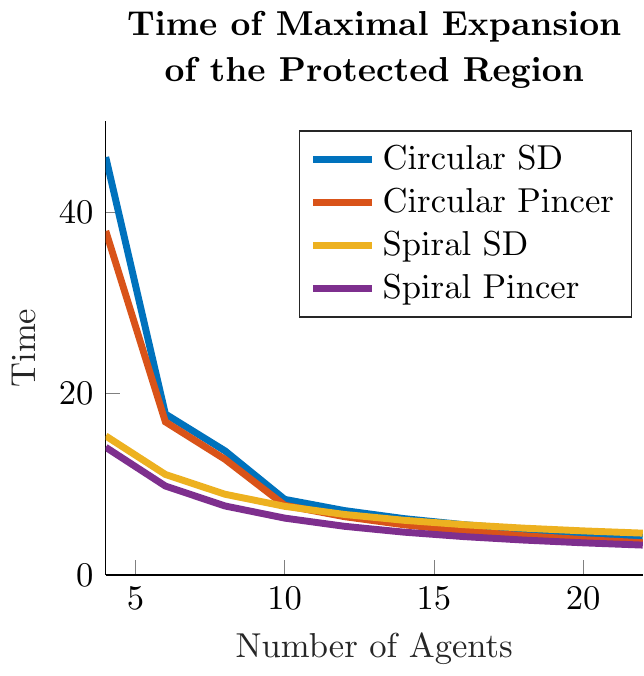} \caption{Zoom in on the time of maximal expansion of the protected region for the circular and spiral defense sweep protocols as well as the spiral and circular defense same-direction sweep protocols. We simulated sweep protocols with an even number of defenders, ranging from $4$ to $22$ defenders, that perform the defense sweep protocols at speeds of $10 V_T$ above the critical speed of $2$ defenders that perform the circular defense same-direction sweep process. The chosen values of the parameters are $r=10$, and $R_{max} = 120$.}
\label{Fig21Label}
\end{figure}

\section{Conclusions and Future Research Directions}
\label{sec:conclusions}
This research studies the problem of guaranteeing defense of an initial region against smart mobile invaders by a swarm of defending agents that act as visual sensors. Invaders are initially located outside a known circular environment which they try to enter without being detected by the defenders. Two novel algorithms that guarantee no intruder enters the region without being detected by a defender team that uses pincer movements between defending pairs are developed and compared to state-of-the-art approaches, proving the superiority of pincer-based defense protocols. Having a speed that exceeds the critical speed that allows defending the initial region, allows the defenders to gradually expand the protected region as well. Numerical and illustrative simulations using MATLAB and NetLogo demonstrate the performance of the proposed algorithms.

A possible extension to this work is to generalize the results for environments with more complex geometries, possibly in the presence of obstacles and apply the pincer expansion protocols in such settings. An additional interesting research avenue is to develop an algorithm that will be robust to failures of defenders and will allow to reorganize the defender team and enable it to continue the defense and expansion tasks with less defenders.



\section*{Appendix A}
\label{sec:Appendix_A}
This appendix provides an analytical computation of the total outward advancement times required in order for a defender team performing the circular defense pincer sweep process to expand the protected region to the maximal defendable area and completes the calculation from section \ref{sec:circular_defense}. Denote the total outward advancement time by $\widetilde{T}_{out}(n) = \sum\limits_{i = 0}^{{N_n} - 2} {{T_{{out}_i}}}$. Equation (\ref{e1100}) is expressed as,
\begin{equation}
\widetilde{T}_{out}(n)= \sum\limits_{i = 0}^{{N_n} - 2} {{T_{{out}_i}}}  = \frac{{\left( {{N_n} - 1} \right)r}}{{{V_s} + {V_T}}} - \frac{{2\pi {V_T}\sum\limits_{i = 0}^{{N_n} - 2} {{R_i}} }}{{n{V_s}\left( {{V_s} + {V_T}} \right)}}
\label{e101}
\end{equation}
The method to calculate $\sum\limits_{i = 0}^{{N_n} - 2} {R_i}$ is developed in Appendix $E$ of \cite{francos2021swarms} and equals,
\begin{equation}
\sum\limits_{i = 0}^{{N_n} - 2} {R_i}  = \frac{{R_0 - {c_2}R_{{N_n} - 2} + ({N_n} - 2){c_1}}}{{1 - {c_2}}}
\label{e102}
\end{equation}
The calculation of $R_{{N_n} - 2}$ is provided in Appendix $B$ of \cite{francos2021swarms}. It is given by,
\begin{equation}
R_{{N_n} - 2} = \frac{{{c_1}}}{{1 - {c_2}}} + {c_2}^{{N_n} - 2}\left( {R_0 - \frac{{{c_1}}}{{1 - {c_2}}}} \right)
\label{e103}
\end{equation}
Replacement of coefficients in (\ref{e103}) results in,
\begin{equation}
\begin{array}{l}
{R_{N_n - 2}} = \frac{{nr{V_s}}}{{2\pi {V_T}}} + {\left( {1 + \frac{{2\pi {V_T}}}{{n\left( {{V_s} + {V_T}} \right)}}} \right)^{{N_n} - 2}}\left( {\frac{{2\pi {R_0}{V_T} - nr{V_s}}}{{2\pi {V_T}}}} \right)
\end{array}
\label{e104}
\end{equation}
exchanging the coefficients in (\ref{e102}) leads to,
\begin{equation}
\begin{array}{l}
\sum\limits_{i = 0}^{{N_n} - 2} {{R_i}}  = \frac{{{R_0}n\left( {{V_s} + {V_T}} \right)}}{{2\pi {V_T}}} - \frac{{{n^2}{V_s}r\left( {{V_s} + {V_T}} \right)}}{{{{\left( {2\pi {V_T}} \right)}^2}}} + \frac{{nr{V_s}\left( {{N_n} - 1} \right)}}{{2\pi {V_T}}} \\ - {\left( {1 - \frac{{2\pi {V_T}}}{{n\left( {{V_s} + {V_T}} \right)}}} \right)^{{N_n} - 1}}\left( {\frac{{n\left( {2\pi {R_0}{V_T} - n{V_s}r} \right)\left( {{V_s} + {V_T}} \right)}}{{{{\left( {2\pi {V_T}} \right)}^2}}}} \right)
\end{array}
\label{e105}
\end{equation}
Replacing the expression for ${\sum\limits_{i = 0}^{{N_n} - 2} {{R_i}} }$ from (\ref{e105}) to equation (\ref{e101}) yields,
\begin{equation}
\begin{array}{l}
\widetilde{T}_{out}(n) = \sum\limits_{i = 0}^{{N_n} - 2} {{T_{ou{t_i}}}} = \\ - \frac{{{R_0}}}{{{V_s}}} + \frac{{nr}}{{2\pi {V_T}}} + {\left( {1 - \frac{{2\pi {V_T}}}{{n\left( {{V_s} + {V_T}} \right)}}} \right)^{{N_n} - 1}}\left( {\frac{{2\pi {R_0}{V_T} - n{V_s}r}}{{2\pi {V_T}{V_s}}}} \right)
 \end{array}
\label{e620}
\end{equation}

\section*{Appendix B}
\label{sec:Appendix_B}
This appendix provides an analytical computation of the total outward advancement times required in order for a defender team performing the spiral defense pincer sweep process to expand the protected region to the maximal defendable area and completes the calculation from section \ref{sec:spiral_defense}. Denote the total outward advancement time by $\widetilde{T}_{out}(n) = \sum\limits_{i = 0}^{{N_n} - 2} {{T_{{out}_i}}}$. Equation (\ref{e702}) is expressed as,
\begin{equation}
\begin{array}{l}
\widetilde{T}_{out}(n) = \sum\limits_{i = 0}^{{N_n} - 2} {{T_{ou{t_i}}} = } \frac{{2r\left( {{N_n} - 1} \right)}}{{{V_s} + {V_T}}} - \frac{{\left( {1 - \lambda } \right)\left( {{V_s} + 1} \right)\sum\limits_{i = 0}^{{N_n} - 2} {{{\tilde R}_i}} }}{{{V_s}\left( {{V_s} + {V_T}} \right)}}
\end{array}
\label{e1087}
\end{equation}
The method to calculate $\sum\limits_{i = 0}^{{N_n} - 2} {R_i}$ is developed in Appendix $E$ of \cite{francos2021swarms} and equals,
\begin{equation}
\sum\limits_{i = 0}^{{N_n} - 2} {\widetilde{R}_i}  = \frac{{\widetilde{R}_0 - {c_2}\widetilde{R}_{{N_n} - 2} + ({N_n} - 2){c_1}}}{{1 - {c_2}}}
\label{e726}
\end{equation}
The calculation of $R_{{N_n} - 2}$ is provided in Appendix $B$ of \cite{francos2021swarms}. It is given by,
\begin{equation}
\widetilde{R}_{{N_n} - 2} = \frac{{{c_1}}}{{1 - {c_2}}} + {c_2}^{{N_n} - 2}\left( {\widetilde{R}_0 - \frac{{{c_1}}}{{1 - {c_2}}}} \right)
\label{e725}
\end{equation}
Replacement of coefficients in (\ref{e725}) results in,
\begin{equation}
\begin{array}{l}
{{\tilde R}_{{N_n} - 2}} = \frac{{2r{V_s}}}{{\left( {1 - \lambda } \right)\left( {{V_s} + 1} \right)}} + \\ {\left( {\frac{{{V_T} + {V_s}\lambda  - 1 + \lambda }}{{{V_s} + {V_T}}}} \right)^{{N_n} - 2}}\frac{{\left( {{R_0} + r} \right)\left( {1 - \lambda } \right)\left( {{V_s} + 1} \right) - 2r{V_s}}}{{\left( {1 - \lambda } \right)\left( {{V_s} + 1} \right)}}
 \end{array}
\label{e503}
\end{equation}
exchanging the coefficients in (\ref{e726}) leads to,
\begin{equation}
\begin{array}{l}
\sum\limits_{i = 0}^{{N_n} - 2} {{{\tilde R}_i}}  = {{\tilde R}_0}\frac{{{V_s} + {V_T}}}{{{V_s}\left( {1 - \lambda } \right)}} - \frac{{2r\left( {{V_T} + {V_s}\lambda } \right)}}{{{V_s}{{\left( {1 - \lambda } \right)}^2}}} - \\ \frac{{{V_s} + {V_T}}}{{{V_s}\left( {1 - \lambda } \right)}}{\left( {\frac{{{V_T} + {V_s}\lambda }}{{{V_s} + {V_T}}}} \right)^{{N_n} - 1}}\left( {{{\tilde R}_0} - \frac{{2r}}{{1 - \lambda }}} \right) + \frac{{2r({N_n} - 2)}}{{1 - \lambda }}
\end{array}
\label{e504}
\end{equation}
Replacing the expression for $\sum\limits_{i = 0}^{{N_n} - 2} {\widetilde{R}_i} $ from (\ref{e504}) to equation (\ref{e1087}) yields,
\begin{equation}
\begin{array}{l}
{\widetilde{T}_{out}}(n) =  \frac{{2r}}{{{V_s} + {V_T}}} - \frac{{{{\tilde R}_0}}}{{{V_s}}} + \frac{{2r\left( {{V_T} + {V_s}\lambda  - 1 + \lambda } \right)}}{{\left( {1 - \lambda } \right)\left( {{V_s} + 1} \right)\left( {{V_s} + {V_T}} \right)}} + \\ {\left( {\frac{{{V_T} + {V_s}\lambda  - 1 + \lambda }}{{{V_s} + {V_T}}}} \right)^{{N_n} - 1}}\frac{{\left( {{R_0} + r} \right)\left( {1 - \lambda } \right)\left( {{V_s} + 1} \right) - 2r{V_s}}}{{{V_s}\left( {1 - \lambda } \right)\left( {{V_s} + 1} \right)}}
\end{array}
\label{e618}
\end{equation}

\bibliographystyle{IEEEtran}
\bibliography{Search_For_Smart_Evaders_TRO}

\begin{thebibliography}{10}
\providecommand{\url}[1]{#1}
\csname url@samestyle\endcsname
\providecommand{\newblock}{\relax}
\providecommand{\bibinfo}[2]{#2}
\providecommand{\BIBentrySTDinterwordspacing}{\spaceskip=0pt\relax}
\providecommand{\BIBentryALTinterwordstretchfactor}{4}
\providecommand{\BIBentryALTinterwordspacing}{\spaceskip=\fontdimen2\font plus
\BIBentryALTinterwordstretchfactor\fontdimen3\font minus
  \fontdimen4\font\relax}
\providecommand{\BIBforeignlanguage}[2]{{%
\expandafter\ifx\csname l@#1\endcsname\relax
\typeout{** WARNING: IEEEtran.bst: No hyphenation pattern has been}%
\typeout{** loaded for the language `#1'. Using the pattern for}%
\typeout{** the default language instead.}%
\else
\language=\csname l@#1\endcsname
\fi
#2}}
\providecommand{\BIBdecl}{\relax}
\BIBdecl

\bibitem{koopman1980search}
B.~O. Koopman, \emph{Search and screening: general principles with historical
  applications}.\hskip 1em plus 0.5em minus 0.4em\relax Pergamon Press, 1980.

\bibitem{stone2016optimal}
L.~D. Stone, J.~O. Royset, and A.~R. Washburn, ``Optimal search for moving
  targets (international series in operations research \& management science
  237),'' \emph{Cham, Switzerland: Springer}, 2016.

\bibitem{rekleitis2004limited}
I.~Rekleitis, V.~Lee-Shue, A.~P. New, and H.~Choset, ``Limited communication,
  multi-robot team based coverage,'' in \emph{Proceeding of the IEEE
  International Conference on Robotics and Automation (ICRA) 2004},
  vol.~4.\hskip 1em plus 0.5em minus 0.4em\relax IEEE, 2004, pp. 3462--3468.

\bibitem{alpern2006theory}
S.~Alpern and S.~Gal, \emph{The theory of search games and rendezvous}.\hskip
  1em plus 0.5em minus 0.4em\relax Springer Science \& Business Media, 2006,
  vol.~55.

\bibitem{vincent2004framework}
P.~Vincent and I.~Rubin, ``A framework and analysis for cooperative search
  using uav swarms,'' in \emph{Proceedings of the 2004 ACM symposium on Applied
  computing}.\hskip 1em plus 0.5em minus 0.4em\relax ACM, 2004, pp. 79--86.

\bibitem{altshuler2008efficient}
Y.~Altshuler, V.~Yanovsky, I.~A. Wagner, and A.~M. Bruckstein, ``Efficient
  cooperative search of smart targets using uav swarms,'' \emph{Robotica},
  vol.~26, no.~4, pp. 551--557, 2008.

\bibitem{mcgee2006guaranteed}
T.~G. McGee and J.~K. Hedrick, ``Guaranteed strategies to search for mobile
  evaders in the plane,'' in \emph{2006 American Control Conference}.\hskip 1em
  plus 0.5em minus 0.4em\relax IEEE, 2006, pp. 6--pp.

\bibitem{shishika2018local}
D.~Shishika and V.~Kumar, ``Local-game decomposition for multiplayer
  perimeter-defense problem,'' in \emph{2018 IEEE Conference on Decision and
  Control (CDC)}.\hskip 1em plus 0.5em minus 0.4em\relax IEEE, 2018, pp.
  2093--2100.

\bibitem{shishika2019team}
D.~Shishika, J.~Paulos, M.~R. Dorothy, M.~A. Hsieh, and V.~Kumar, ``Team
  composition for perimeter defense with patrollers and defenders,'' in
  \emph{2019 IEEE 58th Conference on Decision and Control (CDC)}.\hskip 1em
  plus 0.5em minus 0.4em\relax IEEE, 2019, pp. 7325--7332.

\bibitem{shishika2020cooperative}
D.~Shishika, J.~Paulos, and V.~Kumar, ``Cooperative team strategies for
  multi-player perimeter-defense games,'' \emph{IEEE Robotics and Automation
  Letters}, vol.~5, no.~2, pp. 2738--2745, 2020.

\bibitem{makkapati2019optimal}
V.~R. Makkapati and P.~Tsiotras, ``Optimal evading strategies and task
  allocation in multi-player pursuit--evasion problems,'' \emph{Dynamic Games
  and Applications}, vol.~9, no.~4, pp. 1168--1187, 2019.

\bibitem{chung2011search}
T.~H. Chung, G.~A. Hollinger, and V.~Isler, ``Search and pursuit-evasion in
  mobile robotics,'' \emph{Autonomous robots}, vol.~31, no.~4, pp. 299--316,
  2011.

\bibitem{kumkov2017zero}
S.~S. Kumkov, S.~Le~M{\'e}nec, and V.~S. Patsko, ``Zero-sum pursuit-evasion
  differential games with many objects: survey of publications,'' \emph{Dynamic
  games and applications}, vol.~7, no.~4, pp. 609--633, 2017.

\bibitem{weintraub2020introduction}
I.~E. Weintraub, M.~Pachter, and E.~Garcia, ``An introduction to
  pursuit-evasion differential games,'' in \emph{2020 American Control
  Conference (ACC)}.\hskip 1em plus 0.5em minus 0.4em\relax IEEE, 2020, pp.
  1049--1066.

\bibitem{garcia2019optimal}
E.~Garcia, D.~W. Casbeer, and M.~Pachter, ``Optimal strategies of the
  differential game in a circular region,'' \emph{IEEE Control Systems
  Letters}, vol.~4, no.~2, pp. 492--497, 2019.

\bibitem{garcia2020multiple}
E.~Garcia, D.~W. Casbeer, A.~Von~Moll, and M.~Pachter, ``Multiple pursuer
  multiple evader differential games,'' \emph{IEEE Transactions on Automatic
  Control}, vol.~66, no.~5, pp. 2345--2350, 2020.

\bibitem{agmon2008multi}
N.~Agmon, S.~Kraus, and G.~A. Kaminka, ``Multi-robot perimeter patrol in
  adversarial settings,'' in \emph{2008 IEEE International Conference on
  Robotics and Automation}.\hskip 1em plus 0.5em minus 0.4em\relax IEEE, 2008,
  pp. 2339--2345.

\bibitem{agmon2011multi}
N.~Agmon, G.~A. Kaminka, and S.~Kraus, ``Multi-robot adversarial patrolling:
  facing a full-knowledge opponent,'' \emph{Journal of Artificial Intelligence
  Research}, vol.~42, pp. 887--916, 2011.

\bibitem{francos2021swarms}
R.~M. Francos and A.~M.~Bruckstein, ``Search for smart evaders with swarms of
  sweeping agents,'' \emph{IEEE Transactions on Robotics}, vol.~38, no.~2, pp.
  1080--1100, 2021.

\bibitem{francos2021pincer}
R.~M. Francos and A.~Bruckstein, ``Pincer-based vs. same-direction strategies
  of search for smart evaders by swarms of agents,'' \emph{arXiv preprint
  arXiv:2104.06940}, 2021.

\bibitem{tisue2004netlogo}
S.~Tisue and U.~Wilensky, ``Netlogo: A simple environment for modeling
  complexity,'' in \emph{International conference on complex systems},
  vol.~21.\hskip 1em plus 0.5em minus 0.4em\relax Boston, MA, 2004, pp. 16--21.

\bibitem{francos2019search}
R.~M. Francos and A.~M. Bruckstein, ``Search for smart evaders with sweeping
  agents,'' \emph{Robotica}, vol.~39, no.~12, pp. 2210--2245, 2021.

\end{thebibliography}
\vfill

\end{document}